\documentclass{article}
\usepackage[utf8]{inputenc}
\usepackage{mathtools}
\usepackage{amsmath}%
\usepackage{amsthm}
\usepackage{amsfonts}
\usepackage{amssymb}%
\usepackage{graphicx}
\usepackage{stmaryrd}
\usepackage[T1]{fontenc}
\usepackage[english]{babel}
\usepackage{csquotes}
\usepackage{geometry}
\usepackage{tikz}
\usepackage{algorithm2e}
\usepackage{chemist}
\usepackage{authblk}

\usepackage[skip=3pt,font=scriptsize]{caption}

\newcommand{\shorteqnote}[1]{& & & \text{\smaller\llap{#1}}}

\usepackage{multirow, makecell}
\usepackage{array}
\usepackage{alltt}
\newtheorem{theorem}{Theorem}[]
\newtheorem{proposition}{Proposition}[]
\newtheorem{corollary}{Corollary}[theorem]
\newtheorem{corollary_p}{Corollary}[proposition]

\providecommand{\keywords}[1]
{
  \small	
  \textbf{\textit{Keywords---}} #1
}

\let\Item\item

\newenvironment{alphlist}{%
		\let\item\Item
  \begin{enumerate}%
}{%
  \end{enumerate}}
\newenvironment{romanlist}{%

	\let\item\Item
	\begin{enumerate}
	}{%
	\end{enumerate}}

\author[1]{P\'eter~Boldog}

\affil[1]{
Bolyai Institute, University of Szeged\par 
Aradi v\'ertan\'uk tere 1.\par
Szeged, H-6720, Hungary\par
boldogpeter@gmail.com}

\title{Exact lattice-based stochastic cell culture simulation algorithms
incorporating spontaneous and contact-dependent reactions} 

\begin{document}

\maketitle

\begin{abstract}
In this paper, we address the modeling issues of cell movement and division with a special focus on the phenomenon of volume exclusion in a lattice-based, exact stochastic simulation framework.
We propose a new exact method, called Reduced Rate Method -- RRM, that is substantially quicker than the previously used exclusion method, for large number of cells.
In addition, we introduce three novel reaction types: the contact-inhibited, the contact-promoted, and the spontaneous reactions. To the best of our knowledge, these reaction types have not been taken into account in lattice-based stochastic simulations of cell cultures. These new types of events may be easily applied to complicated systems, enabling the generation of biologically feasible stochastic cell culture simulations.
Furthermore, we show that the exclusion algorithm and our RRM algorithm are mathematically equivalent in the sense that the next reaction to be realized and the corresponding sojourn time both belong to the same reaction and time distributions in the two approaches -- even with the newly introduced reaction types.

Exact, agent-based, stochastic methods of cell culture simulations seem to be undervalued and are mostly used as benchmarking tools to validate deterministic approximations of the corresponding stochastic models.
Our proposed methods are exact, they are easy to implement, have a high predictive value, and can be conveniently extended with new features. Therefore, these approaches promise a great potential.
\end{abstract}


\section{Introduction}
With the emergence of personalized therapeutic approaches, computational drug testing, and even artificial tissue engineering modeling of cell cultures have become a prominent area of mathematical and \textit{in silico} biology. 
Understanding the complex collective behavior emerging from primitive phenomena, such as migration, proliferation, or intercellular communication of individual cells is essential to be productive in these fields. 
Researchers use a wide variety of mathematical and numerical tools to describe the behavior of populations.

Parunak et al. \cite{ParunakABM} distinguishes between equation-based models and the agent-based models. Equation-based models are typically composed of deterministic or stochastic approaches and provide well-established methods for both analytical and numerical treatments of the problem. For example, to determine the evolution of a population over time or to predict certain events. They are also used to prove or rule out specific behaviors of the system: such as extinction of a population, chaotic phenomena, or the absence of periodic behavior, for instance. 

Deterministic models consist of systems of ordinary, partial, or delayed differential equations and are used with great success to predict and analyze the behavior of real-world systems \cite{AllenMathBioBook}\cite{BookKuttlerMuller}. These models usually treat the state variables, as well as the actual individuals in the population model, as if they were continuous quantities (densities or concentrations), although, populations or even chemical solutions are composed of discrete individuals (particles). Moreover, for these models to be predictive, one has to assume that the size of the population is large or apply corrections into the model. We refer to the infamous atto-fox problem and Fowler's article \cite{Fowler} for a detailed discussion and feasible solutions. The time evolution of the model can be obtained with various numerical methods.

When it is important to work with discrete individuals or when the number of individuals is small, the theory of stochastic processes may provide a feasible alternative, especially when stochastic effects play an important role in the evolution of the system. These are mostly Markovian stochastic processes with a discrete state space in discrete or continuous time \cite{AllenStochaBook}\cite{BookKuttlerMuller}\cite{Renshaw}. Analytical treatment of complex stochastic models is usually not trivial, the gold standard among the numerical approaches is Gillespie's stochastic simulation algorithm (SSA) \cite{Gillespie-1}\cite{Gillespie-2} and its variants \cite{CaoEfficient}\cite{TauLeap-1}. We shall adapt this algorithm to the simulation of cell cultures in this article.

Agent-based models (ABMs) or individual-based models \cite{RadekErbanIBM}\cite{ParunakABM}\cite{WilenskyRandABM} are feasible alternatives to the equation based models. ABMs, unlike continuum models, regard every particle as an individual that follows a prescribed set of rules. Information about the system can be obtained by analyzing the collective behavior of the agents with statistical methods.
This technique offers a broader description of individual behavior and accounts for stochastic effects caused by the finite number of agents and their interactions.
As Bonabeau \cite{BonabeauABM} points out: "Individual behavior exhibits memory, path-dependence, and hysteresis, non-markovian behavior, or temporal correlations, including learning and adaptation".
However, the rigorous mathematical analyzis is usually difficult, and simulations can be resource-intensive. ABMs are subject for existential-type proofs, to prove the existence of a certain behavior of the population, even in cases when the formulation of a corresponding equation-based model is not evident \cite{WilenskyRandABM}. Equation-based models have a well-developed mathematical theory, yet agent-based models are mostly expressed as computer routines.

In the case of equation-based population models, we usually decrease the complexity of the investigated phenomena by applying a classification on the individuals according to the trait we intend to examine -- assuming that the differences in this attribute can be averaged out within a class or can be accounted for with a corresponding distribution. 
We usually aim to keep the number of classes, and, hence, the number of parameters required as low as possible. Meanwhile, the model remains of a manageable size, yet it still captures the essence of the modeled phenomena.
In contrast, ABMs are mainly used when the number of these classes is close to the number of members in the population, i. e., the classifying traits are parameter-specific to the individuals and highly influence their behavior. 
Therefore, in these complex cases, we may use ABMs and focus on simulations.

Cells are mesoscopic entities in multicellular organisms. Depending on the scale at which their population is examined and on the complexity of the investigated phenomena, one may find an appropriate approach among the modeling paradigms described above. However, a typical approach is to approximate the behavior of the population with a continuum mean field model \cite{dyson2012macroscopic}, such as a variant of the logistic equation or the Fisher–Kolmogorov equation \cite{DelayDin}\cite{ECMI}, incorporate corrections if spatial effects play an important role \cite{CorrBDM}\cite{IncSpacCorrMeanF}, and analyze the obtained deterministic model with rigorous analytical means. Thus, the underlying ABM play the role of a benchmarking tool. We would like to acknowledge the early, but still significant review by Fredrickson et al. \cite{FredStat} and the excellent review by Charlebois and Balázsi \cite{ModCellPopDyn}.

The development of cell cultures is influenced by a number of important phenomena. 
The two most studied ones may well be cell movement and cell division. 
These two phenomena play notable roles in several important processes, including wound healing, tissue development, or even cancer propagation (Maini and Baker \cite{CollCell}; or for a more experimental viewpoint, see Méhes and Vicsek \cite{MehesVicsek}).
In most cases, approximating cells as point-like particles may not be sufficient, since spatial effects usually affect the development of the cell population \cite{CorrBDM}. 
Volume exclusion is such a significant spatial effect, which expresses that two cells cannot occupy the same volume in space, thus cells are physical obstacles to each other.
Technically, the 'excluded volume' of a cell is the space that is unavailable to other cells in the population as a result of the presence of the first cell.

Probably the most convenient way to incorporate volume exclusion into cell culture models is discretizing the space by an underlying lattice.
Even though there are several well known on-lattice methods: cellular automaton models, cellular potts models, or cellular lattice gas models, we improved a variant of Gillespie's SSA that is extended to work in spatially inhomogenous environment and incorporates volume exclusion.
\\\\
In this work, we propose an extended version of the algorithm originally defined by Baker et al. \cite{CorrBDM}.
We call this the Prompt Decision Method (PDM).
This method introduces new reactions to the modeling framework: contact-inhibited, contact-promoted, and spontaneous reactions, and also broadens the idea of volume exclusion.
Furthermore, we also propose a novel, faster, and computationally more efficient method, the Reduced Rate Method (RRM), that, as we intend to prove, is equivalent to the PDM algorithm. 
Finally, we also provide a more convenient, but also equivalent version of the RRM algorithm, the marginal Reduced Rate Method (mRRM).

Reactions are random events that can occur with a given probability during the life of the simulated cells according to the considerations of the model.
The rate of contact-inhibited reactions of a cell decreases as the number of its free neighbors increases.
Therefore, migration and proliferation naturally fall into this reaction type, and further physiological processes that are inhibited in a crowded environment may also be accounted for in this category.
In a similar manner, the rate of contact-promoted reactions of a cell increases as the number of occupied neighbors increase. On the other hand, the rate of spontaneous reactions is not affected by the number of occupied neighboring lattice sites. We shall explain these concepts in details in \ref{sec:reactions}.

In the remaining part of the Introduction, we summarize Gillespie's SSA for well-stirred chemical systems, then we shortly describe the modified version of the SSA used to simulate the dynamics of moving and proliferating cells on a lattice, proposed by Baker et al. \cite{IncSpacCorrMeanF}. We also point out the limitations of this algorithm in the case of a crowded environment.

\subsection{Gillespie's stochastic simulation algorithms}\label{sec:gil-method}
Gillespie \cite{Gillespie-1}\cite{Gillespie-2} formulated the SSA in order to produce computer-based numerical experiments of chemical reactions. 
It is assumed that chemical species or reactants ($S_1,S_2,\ldots,S_m$) interact via reaction channels ($R_1, R_2,\ldots,R_M$) in a well-stirred reaction chamber.
Time evolution of the amount of reactants ($\underbar{S}(t)=(S_1(t),S_2(t),\ldots,S_m(t))$) is described by a Markov process that is continuous in time ($t\in[0,\infty)$) and discrete in state space ($\underbar{S}(t)\in\mathbb{N}_0^m$). The SSA obtains stochastic realizations of the system by answering two questions \cite{Gillespie-1}: 
\begin{romanlist}
    \item In what time ($\tau$) will the next reaction occur?\label{Q1}
	\item What kind of reaction ($R_1, R_2,\ldots,R_M$) will the next reaction be?\label{Q2}
\end{romanlist}

To answer these questions, Gillespie developed two separate but mathematically equivalent variants of the SSA: the 'direct method' and the 'first-reaction method'.
The key distinction is in the method the algorithms determine which reaction channel to activate. 
It is clear from the formulation \cite{Gillespie-1} that the computational cost of the \textit{direct method} is almost always smaller than the cost of the \textit{first-reaction method}, thus we only review and make use of the \textit{direct method} in this work.
For a comparison of the variants of the SSA, refer to Gillespie's original article \cite{Gillespie-1} and to the work by Cao et al. \cite{CaoEfficient}.

\subsubsection*{The direct method algorithm}
Starting from the vector of initial values $(S_1(0),\dots,S_m(0))\in\mathbb{N}_0^m$, the time series of the state variables $\underbar{S}(t)$ is generated by constantly updating the state variables in properly generated subsequent times ($\tau$) according to the answers to question (\ref{Q1}) and (\ref{Q2}).
To this end, the \textit{reaction probability density function}, $P(\tau, \mu)d\tau$, is defined \cite{Gillespie-1}. That is, at time $t$, the next reaction in the reaction chamber will occur in the differential time interval $(t + \tau, t + \tau + d \tau)$ with the probability $P(\tau, \mu)d\tau$, and will be an $R_\mu$ reaction ($\mu\in {1,\ldots,M}$).

With the procedure of conditioning \cite{Gillespie-1}, the two-variable density function $P(\tau,\mu)$ can be written as the product of two one-variable probability density functions:
\begin{equation}\label{eq:cond}
	P(\tau,\mu) = P_1(\tau) \cdot P_2(\mu|\tau).
\end{equation}
In particular, from \cite{Gillespie-1}, it turns out that $$P(\tau,\mu) = a_\mu \cdot e^{-\tau\cdot a},$$ where the symbol $a_\mu$ $(\mu\in\{1,\dots,M\})$ stands for the so-called propensity function that characterizes reaction $R_\mu$, and may depend on the quality and the actual number of the reactants, their number of combinations or environmental factors, etc. 
In \cite{Gillespie-1} it is assumed that $a_\mu=c_\mu h_\mu$, where $c_\mu$ corresponds to the rate constant of the reaction $R_\mu$ and $h_\mu$ is the number of distinct combinations of reactants participating in $R_\mu$, being present in the reaction chamber at time $t$. 
Although $a_\mu$ and $h_\mu$ clearly depend on time via the number of reactants at time $t$, the time dependence is usually not indicated to simplify the notations. 
Variable $a$ stands for the sum of the propensity functions: $a=\sum_{\mu=1}^M a_\mu$.
In Table \ref{t:combinations} we summarized the formula for $h_\mu$ in case of some important reactant combinations. For detailed explanation consult \cite{Gillespie-1}.

\begin{table}[h!]
\centering
\begin{tabular}{|c|c|}
\hline
Reaction type & Reactant combinations\\
\hline\hline
$S_i\rightarrow$ products & $h_\mu=S_i(t)$ \\
\hline
\multirowcell{2}{$S_i+S_k\rightarrow$ products\\ \scriptsize{($i\neq k$)}}&\multirowcell{2}{ $h_\mu=S_i(t)S_k(t)$}\\
&\\
\hline
$2S_i\rightarrow$ products & $h_\mu=S_i(t)(S_i(t)-1)/2$\\
\hline
\multirowcell{2}{$S_i+S_k+S_l\rightarrow$ products\\ \scriptsize{($i\neq k\neq l\neq i$)}} & \multirowcell{2}{$h_\mu=S_i(t)S_k(t)S_l(t)$}\\
&\\
\hline
\multirowcell{2}{$S_i+2S_k\rightarrow$ products\\ \scriptsize {($i\neq k$)}\\} & \multirowcell{2}{$h_\mu=S_i(t)S_k(t)(S_k(t)-1)/2$}\\
&\\
\hline
$3S_i\rightarrow$ products& $h_\mu=S_i(t)(S_i(t)-1)(S_i(t)-2)/6$\\
\hline
\end{tabular}
\caption{\textbf{Common reactant combinations in chemical reactions} 
A subset of chemical species $\{S_1,\dots,S_n\}$ may react in various combinations $h_\mu$ in the reaction chamber.
The first row corresponds to the decay of species $S_i$. The second and third row correspond to the two possible bimolecular reactions: the reactants are different or the same chemical species, respectively. 
Rows 4, 5, and 6 correspond to the possible trimolecular reaction types in which all three molecules are of different species, two are of the same species, and all of the same species, respectively.}
\label{t:combinations}
\end{table}

Let $P_1(\tau)d\tau$ be the probability that the next reaction will occur between times $t+\tau$ and $t+\tau+ d\tau$, independent of which reaction it might be. Similarly, $P_2(\mu|\tau)$ is the probability that the next reaction will be the $R_\mu$ reaction, given that the next reaction occurs at $t+\tau$. The probability $P_1(\tau)d\tau$ is the marginal distribution for $\tau$, and obtained by summing the reaction probability density function over all possible $\mu$ values:
\begin{equation}\label{eq:P1}
	P_1(\tau)=\sum_{\mu=1}^m P(\tau,\mu)=a\cdot e^{-\tau\cdot a}.
\end{equation}
Substituting $P_1(\tau)$ into (\ref{eq:cond}) and solving for $P_2(\mu|\tau)$ we obtain
\begin{equation}\label{eq:P2}
	P_2(\mu|\tau)=\frac{a_\mu}{a}.
\end{equation}
From which it is clear that $\int_0^\infty P_1(\tau)d\tau=\int_0^\infty a\cdot e^{-\tau\cdot a}=1$ and $\sum_{\mu=1}^{M}P_2(\mu|\tau)=\sum_{\mu=1}^{M}a_\mu/a=1.$
Thus, the so-called sojourn time (or waiting time) $\tau$ between consecutive reactions is an exponentially distributed continuous random variable with parameter $a$, the sum of the reaction propensity functions. And the algorithm randomly selects reaction index $\mu$, which is a discrete random variable, from the discrete probability distribution (Eq. \ref{eq:P2}).

One may generate such $\tau$ from Eq. (\ref{eq:P1}) by selecting $r_1\in(0,1)$ with a continuous uniform distribution and calculate $\tau=1/a\ln(1/r_1)$. 
To generate a reaction index $\mu$ from Eq. (\ref{eq:P2}), one has to select $r_2\in(0,1)$ with a continuous uniform distribution and choose $\mu$ to be the index for which $\sum_{\nu=1}^{\mu-1}a_\nu<r_2a\leq \sum_{\nu=1}^{\mu}a_\nu$ holds. 
This is discussed in details in \cite{Gillespie-1}, and most programming languages have highly optimised built-in functions to execute this task.

Upon choosing reaction $R_\mu$ and $\tau$, the algorithm has to obtain the change vector $(\Delta S_1,\dots,\Delta S_m)\in \mathbb{Z}^m$ that describes the change in the amount of reactants according to $R_\mu$. Finally, both time and the state variables have to be updated.
Gillespie's SSA requires a simple data structure to record the time evolution of the system: a list containing the subsequently appended values $(S_1(t),\dots,S_m(t))$ after each iteration.
With these notations and probability distributions, the direct method SSA can be summarised as follows.

{\smaller
\paragraph{Direct Method Stochastic Simulation Algorithm}
\begin{enumerate}
	\item \textbf{Initialization:}
	\begin{alphlist}
	    \item Initialize data structure that stores time series.
		\item set $t\leftarrow0$,
		\item set initial values $(S_1(0),\dots,S_m(0))\in\mathbb{N}_0^m$, 
		\item prescribe halting conditions $H$.
	\end{alphlist}
	\item \textbf{Calculate propensity functions} $a_\mu$ for all $\mu\in\{1,2,\ldots,M\}$, calculate $a$.
	\item \textbf{Decide when the next reaction will occur:} choose $\tau$ according to Eq.~(\ref{eq:P1}),
	\item \textbf{Reaction selection:} choose $\mu$ according to Eq.~(\ref{eq:P2}),
	\item \textbf{Update time and state variables:}
	\begin{alphlist}
		\item set $t\leftarrow t+\tau$
		\item obtain the change $(\Delta S_1,\dots,\Delta S_m)\in \mathbb{Z}^m$ in number of reactants according to $\mu$,
		\item change the number of molecules according to $$(S_1(t+\tau),\dots,S_m(t+\tau))=(S_1(t),\dots,S_m(t))+(\Delta S_1,\dots,\Delta S_m).$$
	    \item Update data structure.
	\end{alphlist}
	\item \textbf{Halt if $H = True$ else continue the process with Step (2).}
\end{enumerate}
}

\subsubsection*{Strengths and limitations of the SSA}\label{strength_lim}
It became apparent that Gillespie's approach is particularly beneficial in case of certain biological systems when only a small quantity of reactants are present in the reaction chamber, thus stochastic fluctuations play an important role in the modeled physiological process. 
A well-known example of a process of this type is intracellular gene regulation \cite{RNSStoch}, where only a few copies of the corresponding regulatory molecules are present in the cell \cite{SmallNumStoch}.
In addition to chemical kinetics \cite{erban_chapman_2020}\cite{Persp}\cite{ReactKin}, the SSA has been proven useful for studying epidemic spread \cite{Allen}\cite{PropMtx} and describing ecological phenomena \cite{GEM}, among a variety of further applications.

Since the SSA produces the time evolution of the system by recording every single reaction, it may be time consuming for systems that are characterized by large state space or contain reactions that have a high rate compared to other reaction rates in the system. 
There have been several attempts to solve these problems, the most influential of these may be the one by Gibson and Bruck \cite{GibsonBruck}. 
Their Next Reaction Method (NRM) is an extension of Gillespie's first reaction method. They introduced the so-called reaction dependency graph. 
The reaction dependency graph specifies which propensity function $a_\nu$ is to be changed after the selected reaction $R_\mu$ is executed. Cao et al. \cite{CaoEfficient} made a detailed comparison of the DM, the FRM, and the NRM, and found that "even with the best data structure, the NRM is less efficient than the DM, except for a very specific class of problems." The NRM handles a complicated data structure and according to Schwehm \cite{Schwehm}: "the simulator engine spends most of its execution time maintaining the priority queue of the tentative reaction times".

Cao et al. (2004) \cite{CaoEfficient} optimized the algorithm by using the clever observation that in case of robust systems, the reactions are usually multiscale, which means that some of the reactions fire more frequently than others. Therefore, by sorting the indices of the reactions in decreasing order based on how often they fire, they were able to speed up the reaction selection step. Moreover they adopted the core idea of the NRM \cite{GibsonBruck} and only updated the propensities of the reaction channels that were affected by the last reaction. The idea of separating fast and slow reactions was further improved in the \textit{slow-scale stochastic simulation algorithm} by Cao et al. (2005) \cite{SlowSSA}.
Finally, we would also like to mention Gillespie's tau-leaping method \cite{TauLeap-1} and its improved versions \cite{Tau-2}\cite{Tau-3}. It is an approximate method analogous to the Euler method in deterministic systems, and it executes several reactions in a time interval before updating the propensities.

\subsection{The lattice-based variant of the SSA}\label{sec:original_LBSSA}

Baker et al. introduced a variant of Gillespie's SSA to simulate a population of cells on a lattice \cite{CorrBDM}. The algorithm incorporates volume exclusion and it was used as a tool to produce numerical experiments to benchmark the mean field approximations of the authors. 

Spherical cells (with uniform size of diameter $d$) are placed on a finite, regular, square lattice of two dimension with grid constant $d$ and von Neumann neighborhood. A cell may engage in three events: movement ($R_1$), division ($R_2$) or death ($R_3$) with the corresponding rates $r_1, r_2, r_3$, in order. In every iteration, the corresponding propensity functions $a_\mu=r_\mu N$, $\mu\in\{1,2,3\}$, $a=a_1+a_2+a_3$ are calculated, where $N=N(t)$ is the number of cells in the population at $t$.

A random waiting time is generated from the distribution $\tau\sim Exp(a)$ and the time of the system is updated: $t\leftarrow t+\tau$. A target cell is randomly selected from the population with discrete uniform distribution. 
Then, the next possible reaction is selected from the distribution $P(\mu|\tau)=a_\mu/a$ and executed on the target cell with the following constraints: 
\begin{itemize}
    \item in the case of movement, the cell moves to a randomly selected neighboring lattice site if it is empty; if the lattice site is occupied, the reaction is rejected,
    \item in the case of division, the cell places a daughter cell on a randomly selected neighboring lattice site if it is empty and the population size is increased by 1; if the lattice site is occupied, the reaction is rejected,
    \item in the case of death, the cell dies without any particular restriction, in this case the lattice site of the cell is turned empty and the size of the population is decreased by 1.
\end{itemize}

Thus, whenever movement or proliferation cannot be executed because of the volume exclusion, we reject the selected reaction. 
However, we always keep the generated sojourn time $\tau$.
We will not present the algorithm in any more detail here, as we will soon introduce an extended version of it.

A major limitation of the approach is that the target cell, that is about to take action in a given iteration of the simulation, is selected randomly from all the cells in the population with a discrete uniform distribution. Thus, there is a chance that the algorithm selects a cell that is unable to move or divide. In fact, with nonzero proliferation rate and zero death rate the population grows monotonically. On a finite lattice the probability of choosing a cell that cannot engage in any of the possible events grows as the cells would have less empty adjacent sites over time.

For an extreme example, suppose that the grid is finite of size $K=n_1\cdot n_2$ with a periodic boundary condition, thus the maximum number of cells on the lattice is $K$, and suppose that the current number of cells is $K-1$ while the death rate is zero. In such a setting, it is easy to show that the selected event is expected to be rejected $K-1$ times before acceptance.
Indeed, in these cases, i. e., in a crowded environment, it is computationally intensive to run the simulation, not only because \textit{"nothing seems to happen"} most of the time, but also because generating the random waiting time from an exponential distribution requires calculation of a logarithm, which is a rather demanding task. 
Clearly, in this framework we cannot just choose the cell 'wisely' or choose an empty adjacent site 'on purpose' to speed up the simulation, as we are also sampling the distribution of the waiting time. 
Therefore, the number of times we discard a reaction is just as important as the reaction itself. 

In our reduced rate method we are using another approach to define the reactions and the corresponding propensities. In addition, the algorithm has a different basis, so that all the reactions that are drawn are always executed.

Next, we present the prompt decision method (PDM) that applies the volume exclusion principle directly and incorporates the newly defined reactions. 
Then, we introduce our other novel approach, the reduced rate method (RRM), which applies volume exclusion in an indirect way. 
After that, we intend to prove that in a given state of the system the \textit{time until the next reaction to be realized} and also \textit{the next reaction to be realized} are drawn from the same time and event distribution in both the PDM and the RRM algorithms, therefore, they are mathematically equivalent and are interchangeable.
Finally, we experiment with some toy models to illustrate the potential in the newly incorporated contact dependent reactions.
We also provide the corresponding Python codes to illustrate the strength and the flexibility of the new approach.

\section{Methods}\label{met}

\subsection{Lattice discretization and state space}\label{sec:lattice_state}
All of the presented methods make the following assumptions: consider spherical cells of equal size, with diameter $d$, place the cells on a 2-dimensional square lattice with grid constant $d$ and size $(n_1\times n_2)$ where $n_1, n_2\in \{1,2,\dots\}$ are finite or infinite. 
Thus, the lattice has unit length spacing, allowing us to compare various cell lines with their characteristic cell sizes. We use the terms grid and lattice interchangeably. Also the terms neighbor, grid neighbor and adjacent site refer to the same concept.

We assume that the cells use chemical or mechanical mechanisms to sense each other's presence – specifically, whether or not their immediate grid neighbors are occupied. 
In the case of a finite lattice, we use a reflective boundary condition: the boundaries of the grid are perceived by the cells as occupied neighbors. A cell has only free and occupied neighbors: an adjacent site that is not free is occupied and vice versa.
During our toy models, we assume von Neumann neighborhood, thus a given cell has four adjacent neighbors on the lattice, $\eta=4$ (up, down, left, right).
Consequently, a cell may have $j\in\{0,\dots,\eta\}$ free and $(\eta-j)$ occupied neighbors. From now on, index $j$ will always denote the number of free adjacent sites of the concerned cell or cells.
The method we shall discuss may be applied immediately to three dimensions, other lattice types, other neighborhood types (see Fig.\ref{fig:lattice}), and even to lattices with special shapes. 

\begin{figure}[h]
    \centering
    \includegraphics[width=.75\textwidth]{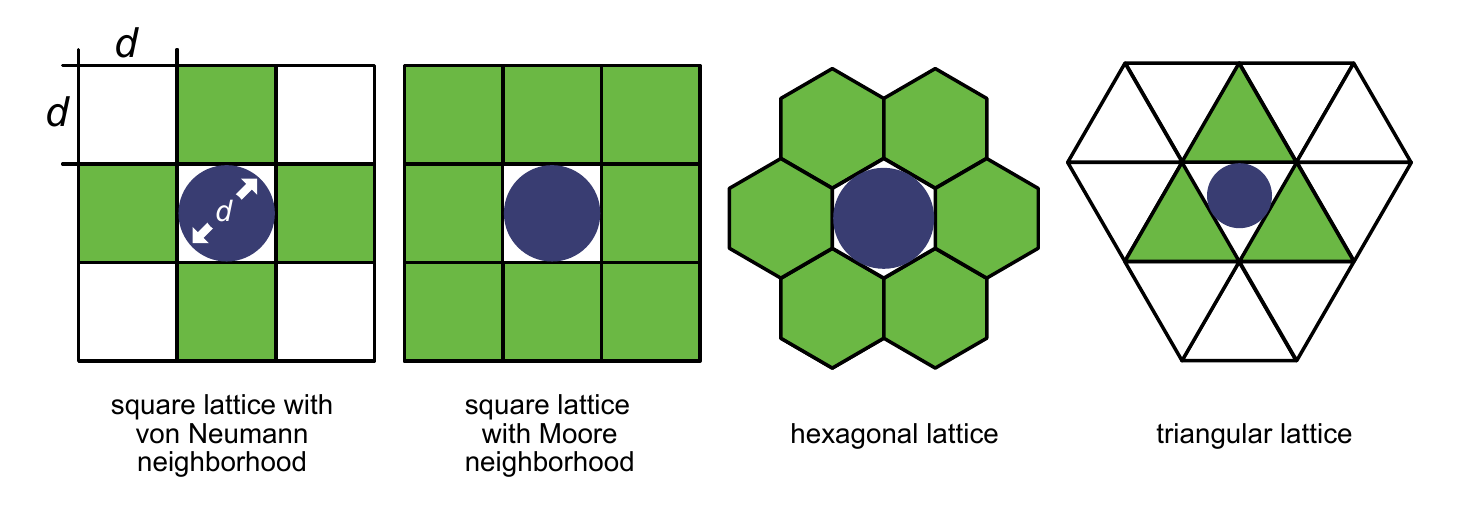}
    \caption{\textbf{Common lattice types in two dimensions.} Square and hexagonal lattices are widely used in lattice-based cell population models, triangular lattices are mostly used in lattice gas models. Each lattice contains one circular cell (purple) and the adjacent lattice sites are shaded green. The number of neighboring sites in these types are $\eta= 4, 8, 6, 3$ in order. Thus a cell may have $j\in\{0,\dots,\eta\}$ free and $(\eta-j)$ occupied adjacent sites. Throughout the paper, we use a square lattice with von Neumann neighborhood (leftmost figure). In this case we assume that the cell has diameter $d$ matching with the grid constant of the lattice as indicated on the figure.} 
    \label{fig:lattice}
\end{figure}

We assume that the principle of volume exclusion applies for all cells: there can be at most one cell per lattice site under the above conditions. 
Thus, by volume exclusion, the capacity of a finite grid with size $(n_1\times n_2)$, where $n_1,n_2<\infty$, is $K=n_1\cdot n_2$. 
Which means that the lattice can support at most $K$ cells.

\paragraph{State space} 
Let $t\in[0,\infty)$ represent time, and $N_j(t)$ denote the number of cells in the population with $j$ free neighbors at time $t$. Then let $N=N(t)=\sum_{j=0}^4N_j(t)$ be the total number of cells in the grid at time $t$ (the time dependence is usually not denoted for simplicity).

We introduce the state vector 
\begin{equation}
    \underbar{X}(t) = (C_1(t),\dots, C_N(t), S_1(t),\dots,S_m(t))
\end{equation}
where $C_i(t)$ is the coordinates of the cell with index $i\in\{1,\dots,N\}$ at time $t$ on the lattice.
By volume exclusion, it is clear that no two cells have the same coordinates at $t$. 
Scalar variables $S_s(t)\in\mathbb{N}_0$ ($s\in\{1,\dots,m\}$) store the amount of chemical species present in the cell space at time $t$. In order to ease the notation we refer to the state of the cells and the state of chemicals the following way:
$$\underbar{C}=(C_1(t),\dots, C_N(t)) \text{ and } \underbar{S}=(S_1(t),\dots,S_m(t)).$$
Fig.\ref{fig:snapshot} shows the time evolution of a cell population along with a snapshot of the lattice configuration at time $t*$.
\begin{figure}[h]
    \centering
    \includegraphics[width=.7
    \textwidth]{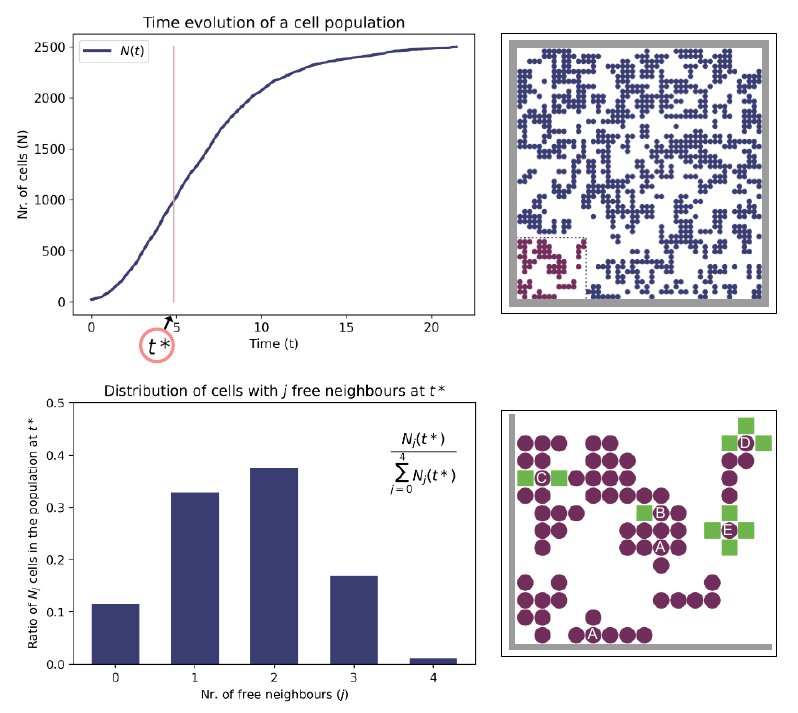}
    \caption{\textbf{Time evolution of a cell culture.} The \textit{top left figure} shows the time evolution of a proliferating cell population on a lattice with carrying capacity $K=2500$. In this experiment we chose the proliferation rate to be 1 and we assumed that no other reaction can happen to the cells. The \textit{bottom left figure} shows a snapshot of the ratio of cells with $j$ free neighbors in the population at $t*=4.8$ (see the formula in the figure). In this state the toal number of cells is $N(t*)=1000$ and very few cells -- only 1.1\% of the population -- have $j=4$ free adjacent sites. In fact, most of the cells have $j=2$ neighbors at $t*$. The \textit{top right figure} shows a snapshot of the lattice configuration at $t*$ and the \textit{bottom right figure} is an enlarged part of the highlighted area of the whole lattice. On this latter figure we selected some cells: A, B, C, D, and E having $0\dots4$ free adjacents sites, in order. Their corresponding free neighbors are colored green and the border of the lattice is marked in gray. We assume that the cells perceive the boundary as an occupied neighbor -- thus both 'A' cells have zero neighbors at $t*$.}
    \label{fig:snapshot}
\end{figure}

\subsection{Reactions}\label{sec:reactions}
We use the term 'reactions' in the sense it is used in chemical modeling. We also refer to movement, division or cell death as reactions. We will adopt the usual notation to denote the reaction index by $\mu\in\{1,\dots,M\}$.
Our goal is to provide a generalized and therefore, a rather flexible modeling framework for cell culture simulation that can be utilized to capture exceptionally complex mechanisms occurring in cell cultures. In order to achieve this, any number and combination of the reactions described in this section can be included in the model the researcher is working on, which covers a multitude of possible activities in a modeled cell culture with very few technical restrictions.

It is important to keep in mind that selecting the reaction and selecting the target cell taking the action are executed in substantially different ways in the PDM and RRM algorithms.
At this point, we only focus on what happens to the state variables if a reaction is '\textit{successful}' in the sense that it is selected and \textit{can be} executed.
The differences will be highlighted in the following sections, when we introduce the methods.
Next, we describe the reactions (Table \ref{t:reactions}.) that are introduced into the model during the prompt decision and reduced rate methods.

\begin{table}[h!]
\centering
{\smaller
\begin{tabular}{ |l|c||c| }
 \hline
 \multicolumn{2}{|c||}{Reaction} & In a finite environment...\\
 \hline \hline
 
\multirowcell{4}{Contact-inhibited} & Movement ($\otimes$) & \multirowcell{4}{the per capita rate decreases \\as the population grows,}\\
& Division ($\otimes$) & \\
& Cell death &\\
& Cellular biochemical reaction &\\
\hline

\multirowcell{3}{Spontaneous}& Cell death &\multirowcell{3}{the per capita rate is \\independent of pop. size,}\\
& Cellular biochemical reaction &\\
& Non-cellular biochemical reaction &\\
\hline

\multirowcell{2}{Contact-promoted}& Cell death &\multirowcell{2}{the per capita rate increases \\ as the population grows.}\\
& Cellular biochemical reaction &\\
\hline
\end{tabular}
}
\caption{\textbf{Common types of reactions that may occur in a cell culture}. $\otimes$ indicates volume exclusion. The reactions defined in this table are intended to provide a modeling framework for the exact simulation of cell cultures, that can be easily adapted to the problem the reader is considering. A prominent aspect of this work would be to show that, using these types of reactions, the PDM and RRM algorithms are mathematically equivalent. For detailed explanations and possible interpretations see Sec. \ref{sec:reactions}.}
\label{t:reactions}
\end{table}

\subsubsection{Non-biochemical reactions}
We define the contact-dependent and contact-independent movement, proliferation, and death of the cells as non-biochemical reactions.

\paragraph{Contact-inhibited reactions: movement, proliferation, and cell death}
Contact-inhibited reaction rates of a cell decreases as the number of its free neighbors decreases.
Typical examples are the volume excluding reactions, movement and proliferation, but we shall define two other types of such reactions as well.
First, consider the movement and division of the cells. 

Cells may perform random walk on the grid with a motility rate $r_{\text{mot}}$ per unit time, moving to another adjacent lattice site or divide with a proliferation rate $r_{\text{div}}$ per unit time giving rise to another agent, such that the principle of volume exclusion applies.
Consequently, a movement rate of $r_{\text{mot}}$ and a division rate of $r_{\text{div}}$ can only be measured for a sufficiently sparse population of cells. 
In a dense population, the cells are expected to interfere with each other, thus lower $r_1$ and $r_2$ rates are observed. 

After selecting the movement or division reaction, a target cell is selected from the population. 
If the reaction is a \textit{successful} movement, then the cell is moved to its selected free adjacent position on the lattice and its coordinates are updated in $\underbar{X}(t)$.
If the reaction is a \textit{successful} cell division, the target cell places a daughter cell ($X_{N+1}$) onto a selected free adjacent site on the lattice and the new cell along with its coordinates is appended to the state vector $\underbar{X}(t)$.

Cell-cell communication with local cellular signals are crucial factors to maintain a cell culture. 
In fact, abandoned cells that do not respond to the environment are destined to die in a healthy tissue. 
\textit{contact-inhibited cell death} may be a model for death due to the lack of these local effects.
In case of a cell death reaction, the target cell is removed from the lattice and its coordinates are deleted from the state vector $\underbar{X}(t)$.

\paragraph{Spontaneous (contact-independent) reaction: cell death}

It is also assumed that cells can be involved in events that are not influenced by the number of free adjacent sites. Such a reaction is, for example, \textit{spontaneous cell death}, which is a standard model of naturally occurring death events due to the finite lifespan of the individuals.
Spontaneous cell death affects all cells equally, regardless of their adjacencies, and when it occurs, the cell that is randomly selected from the whole population would be removed from the lattice, and also from the state vector.

\paragraph{Contact-promoted reaction: cell death}

The rate of contact-promoted reactions increase as the environment gets more crowded. 
Thus, contact-promoted cell death may serve as a model for death due to competition for resources, such as oxygen or glucose, between the cells. 
For example, in tumor pathology, it is a well known effect that the core of some tumors suffer from hypoxia and substrate deprivation due to the lack of sufficient vascularization.
Contact-promoted death may be a good candidate to model death in such cancer models. 
In case of a \textit{successful} contact-promoted death, the target cell is removed from the lattice and also from the state vector. 

\subsubsection{Biochemical reactions}
These reactions include (bio)chemical reactions that do or do not involve the cells themselves, called cellular and non-cellular biochemical reactions respectively. Examples of these reactions would be cell-cell communication, product synthesis (proteins, nucleic acids, hormones, signaling factors, toxins etc.), or substrate consumption. 

\paragraph{Cellular biochemical reactions}
We assume that a subset of the reactants $\{S_1,\dots,S_m\}$ in reactant combination $h_\mu$ and a cell $X_i(t)$ may interact and engage in \textit{cellular biochemical reactions}. 
These reactions may be contact-dependent (inhibited or promoted) or spontaneous reactions. 
Cellular biochemical reactions may have four types of outcomes: i) cell movement, ii) division, iii) death or iv) a $(\Delta S_1,\dots,\Delta S_m)\in\mathbb{Z}^m$ change in the number of the reactants. 
The state space has to be updated according to the ways we discussed so far in this section or in Sec. \ref{sec:gil-method}.

\paragraph{Non-cellular biochemical reactions}
A subset of the reactants $\{S_1,\dots,S_m\}$ may also participate in spontaneous \textit{non-cellular chemical reactions} in reactant combination $h_\mu$. Non-cellular biochemical reactions are independent from the cells and only depend on the reactants.
These reactions may result in a $(\Delta S_1,\dots,\Delta S_m)\in\mathbb{Z}^m$ change in the number of the reactants. The state space has to be updated according to the ways we explained in Sec. \ref{sec:gil-method}.

Note that if we exclude all other reactions, including movement and cell division, and keep only the non-cellular chemical reactions, the cell culture simulation algorithm reduces to Gillespie's SSA.

\subsubsection{Some restrictions concerning the reactants}
We intend to show that the PDM and RRM algorithms are mathematically equivalent ways to obtain realizations of a given stochastic process concerning a cell culture.
In order to achieve this, we must apply some restrictions on the reactants and the state space, limiting the algorithms for the sake of this proof. 
We assume that the state of the cells is only determined by their individual positions on the lattice (and its neighborhood), and they do not differ in any other way. Also, reactants $S_1,\dots,S_m$ do not diffuse in the intercellular space. 

These conditions restrict the possibilities enabled by the RRM algorithm, but these conditions are necessary for a direct comparison and proof of equivalence between the RRM and the more limited PDM approaches. A model defined using the RRM algorithm without these restrictions may be impossible to implement in the PDM approach.

Under these confined conditions, we may still implement a multitude of biological mechanisms: the reactants may be secretions from an internal or external secretory gland, which are not released into the intercellular space but into the bloodstream or into a lumen. Or we may also assume homogenous mixing of the reactants in the cell space so that the products of the cellular or non-cellular chemical reactions reach all cells equally. This can be easily achieved in an experimental setup where we assume a layer of fluid over the cells on a 2D surface, homogenized by continuous mixing.

\section{Cell Culture Simulation Algorithms}
In this section, we introduce the PDM and the RRM algorithms in details. We also obtain an equivalent, but simplified version of the RRM algorithm, the marginal RRM (mRMM), that is expected to be easier to implement.
In order to ease the further notations, for every reaction $\mu\in\{1,,\dots,M\}$ we introduce 
\begin{equation}\label{eq:r}
r_\mu=c_\mu h_\mu.   
\end{equation}
Where $c_\mu$ is the reaction rate of reaction $\mu$. We have to adjust the definition of $h_\mu$, as there may be reactions in the Cell Culture Simulation Algorithms in which chemicals $S_1,\dots,S_m$ are not involved.
{\small
\begin{equation}\label{eq:h_mu}
    h_\mu = \left\{
        \begin{array}{ll}
        1 & \quad \text{if chemical species are not involved in  $\mu$} \\
        \text{Nr. of distinct reactant combinations}& \quad \text{if chemical species are involved in  $\mu$}
    \end{array}
    \right.
\end{equation}
}

See Table \ref{t:combinations} and \cite{Gillespie-1} for some reactions and the corresponding $h_\mu$. 
Note that this notation slightly differs from Gillespie's original notation, where $c_\mu h_\mu$ would be the propensity function corresponding to reaction $\mu$. 
In the cell culture simulation algorithms most reactions depend on the state $\underbar{C}$ of the cells, this is what we are intended to emphasize with this distinguished notation. 
Basically, $a_\mu=r_\mu\cdot f(\underbar{C})$, where the definition of the function $f:\underbar{C}\rightarrow[0,\infty)$ is essentially different in the PDM and the RRM formulation.

Given that the two simulation methods differ not only in the algorithm but also in the definition of the reactions, we will denote reactions, events, and variables specific to the PDM with a hat $(\hat{\bullet})$ symbol.
In case of the RRM, we do not use any specific notation. 
Some variables, events, and their probabilities depend only on the state of the system: e.g. the number of cells $N_j$ with $j$ free neighboring lattice sites, the reaction rates $c_\mu$,  reactant combinations $h_\mu$ etc.
In a given state their values are the same in both methods, so we use the same notations for these in the algorithms. 

\subsection{The Prompt Decision Method (PDM)}

The PDM is an extended version of the stochastic simulation algorithm introduced by Baker et al. \cite{CorrBDM}. 
It is a classical acceptance-rejection sampling method that samples the two-variable density function $P(\tau, \mu)$ given the state $\underbar{X}(t)$ of the system, that contains information of the number $N(t)$ and positions $C_i(t)$ of the cells, as well as the amount of the interacting chemical species $S_s(t)$.
Thus, the PDM directly incorporates the principle of volume exclusion by rejecting reactions that are impossible due to the configuration of the cells on the lattice.

\subsubsection*{Propensities and probabilities}

For the reactions $\hat{R}_\mu$ of types we discussed in Sec. \ref{sec:reactions} (also see Table \ref{t:reactions}) at time $t$, state $\underbar{X}(t)$, we define propensity functions $\hat{a}_\mu$ and their sum $\hat{a}$ the following way:
\begin{equation}\label{eq:prop_ex}
\hat{a}_\mu = r_\mu N, \hspace{0.5cm}
\hat{a}=\sum_{\mu=1}^M\hat{a}_\mu=N\sum_{\mu=1}^Mr_\mu.
\end{equation}
Notice that at first glance non-cellular chemical reactions do not fit in this formulation, since by definition, their propensities only depend on $\underbar{S}(t)$ and do not depend on the state of the cells. To overcome this we may think of the number of cells as a scaling factor $1/N$ to the reaction rate constant. Then, for a non-cellular chemical reaction $\mu$ we have $$\hat{a}_\mu=c_\mu h_\mu=\frac{c_\mu h_\mu}{N} N=r_\mu N.$$
With this notation, non-cellular chemical reactions reduce to a special case of the spontaneous type cellular biochemical reactions.
Let $\hat{P}_1(\hat{\tau})$ denote the probability density function of the sojourn time until the next \textit{possible} event and $\hat{P}_2(\mu|\hat{\tau})$ denote the probability mass function of the next possible event:

\begin{equation} 
\label{eq:pmf_ex}
\hat{P}_1(\hat{\tau})=\hat{a}e^{-\hat{a}\hat{\tau}},\hspace{0.5cm}\hat{P}_2(\mu|\hat{\tau})=\hat{a}_\mu/\hat{a}, \hspace{0.2cm}\mu\in\{1,\dots,8\}.
\end{equation}

\subsubsection*{Assumptions in the PDM algorithm}

The algorithm is derived from the following assumptions:
\begin{itemize}

    \item \textbf{Selecting the next possible reaction:} from $\hat{P}_2(\mu|\hat{\tau})$ (Eq. \ref{eq:pmf_ex}) we draw reaction index $\mu$. We will use the notation $\hat{R}_\mu$ for a selected reaction in the PDM algorithm. It is easy to see that the probability of selecting $\hat{R}_\mu$ does not depend on the actual number of cells in this approach (cf. Eq. \ref{eq:prop_ex} and \ref{eq:pmf_ex}).
    
    \item \textbf{Target cell selection:} a target cell is selected from the population with a discrete uniform distribution, without taking the position of the cells into account. Let $\hat{E}_j$ denote the random event that we select a target cell that has $j$ free neighbors.
    
    \item \textbf{The execution of the reaction depends on the neighborhood of the target cell}, the spontaneous reactions are always executed, contact-inhibited or contact-promoted events may be discarded depending on the neighborhood of the target cell. A more detailed explanation of this acceptance-rejection step is about to come later. During one iteration we make trials to decide if the selected possible reaction can be the next reaction to be realized.
    
    \item \textbf{Selecting the waiting time:} as we may discard the selected \textit{possible reactions} several times during one iteration, we introduce $l\in\{1,2\dots,\}$ to index the trials during one iteration. Thus, in a given iteration for the $l$-th trial we generate sojourn time $\hat{\tau_l}\sim Exp(\hat{a})$ (Eq. \ref{eq:pmf_ex}) corresponding to the next \textit{possible} event.\\
    For the sake of clarity (and the convenience of the proof), we define the variable $\hat{T}$, into which we sum the sojourn times generated for the \textit{possible events} during one iteration.
    Accordingly, the time elapsed between two consecutive \textit{realized events} is
    \begin{equation}
    \label{eq:T}
    \hat{T}=\sum_{l=1}^k\hat{\tau}_l,
    \end{equation}
    given that we were able to execute the reaction in the $k$-th attempt (i. e., $k-1$ possible events were discarded in this case). Naturally $k$ is a discrete random variable that is geometrically distributed, as we shall see later.
    In practice, book the intervent time $\hat{\tau}_l$ right after generating it: $\hat{T}\leftarrow \hat{T}+ \hat{\tau}_l$.
    \item \textbf{During the decision step} at the $l$-th attempt
    \begin{itemize}
        \item In case of contact-inhibited reactions, randomly select a neighboring site of the target cell from its neighbors with discrete uniform distribution. If the site is \textit{occupied}, or it is outside the border, reject the reaction. If the site is \textit{empty} execute the reaction and increase the system time: $t\leftarrow t+ \hat{T}$ and set $\hat{T}=0$.
        \item In case of spontaneous reactions execute the reaction on the selected cell and increase the system time: $t\leftarrow t+ \hat{T}$ and set $\hat{T}=0$.
        \item In case of contact-promoted reactions, randomly select a neighboring site of the target cell from its neighbors with discrete uniform distribution. If the site is \textit{occupied}, or it is outside the border, execute the reaction, increase the system time: $t\leftarrow t+ \hat{T}$ and set $\hat{T}=0$. If the site is \textit{empty}, reject the reaction.
    \end{itemize}
\end{itemize}
The steps of the algorithm do not follow the order of the assumptions, we chose this order to streamline the explanation. We solely introduced the variable $\hat{T}$ to make the proof more convenient. In practice, it is not necessary to occupy memory space for $\hat{T}$ and waste machine time by updating it. One may simply increase the system time right after generating the waiting time.

\subsubsection*{Data structure}

From a practical point of view, the implementation of $\underbar{X}(t)$ is a list or an array.
However, we need a data structure that allows us to easily determine whether a lattice site $(x_*,y_*)$ is empty or occupied ($x_*\in \{1,\dots, n_1\},y_*\in \{1,\dots,n_2\}$). Searching in arrays or lists typically has $O(N)$ time complexity, thus it is more convenient to use the underlying lattice to book-keep the occupied lattice sites. 

Hence, we define the matrix $L\in\{0,1\}^{n_1\times n_2}$ for $n_1,n_2<\infty$, where 0 and 1 stand for an empty and an occupied site, respectively. Checking the value of an element $L[x_*,y_*]$ usually has a constant time complexity, $O(1)$. Using $\underbar{X}(t)$ and $L$ simultaneously has two advantages. First, randomly selecting the target cell from the list $\underbar{X}(t)$ is of $O(1)$, but it can be very time consuming to randomly select a target cell from $L$. Second, with using $L$ we may obtain the occupancy of a lattice site in constant time. These assumptions hold for most programming languages, e. g., for Python that we shall use to implement our toy models.
With the introduced notations, the Prompt Decision Method is as follows.

{\smaller
\paragraph{Algorithm Prompt Decision Method}
\begin{enumerate}
	\item \textbf{Initialization:}
	\begin{alphlist}
    	\item set $t\leftarrow 0$, prescribe halting conditions $H$,
    	\item initialize the lattice $L$ with size $n_1\times n_2$, \\
    	place $N(0)\leq n_1\cdot n_2$ cells on the lattice according to essay,\\ store location of cells and amount of chemical species in $\underbar{X}(t)$
	\end{alphlist}
	\item \textbf{Calculate propensity functions} $\hat{a}_\mu$ for all $\mu\in\{1,\dots,8\}$ according to Eq. (\ref{eq:prop_ex}) and \textbf{set} $\hat{T}=0$.

	\item \label{step:gentime}\textbf{Generate intervent time:} choose $\hat{\tau}\sim Exp(\hat{a})$ (Eq. \ref{eq:pmf_ex}) and
	\textbf{set} $\hat{T}\leftarrow \hat{T}+\hat{\tau}$.
	
	\item \textbf{Reaction selection:} choose $\mu$ according to $\hat{P}_2(\mu|\hat{\tau})$ (Eq. \ref{eq:pmf_ex}).
	
	\item \textbf{Cell selection:} draw target cell $X_k(t)$ from the population $\underbar{X}(t)$ with discrete uniform distribution.
	
	\item \textbf{Decision step:}
	  \begin{itemize}
	    
        \item In case of a non-cellular chemical reaction execute the reaction on the concerned chemical species.\\
        Update data structure. Go to Step \ref{step:updatetime}
	  
        \item In case of a spontaneous reaction execute the reaction on target cell and update data structures.\\
        Go to Step \ref{step:updatetime}.
        
        \item In case of a contact-inhibited reaction randomly select a neighboring site of the target cell from its adjacent sites, with discrete uniform distribution. 
        \begin{itemize}
            \item If the target site is \textit{occupied}, or it is outside the border, reject the reaction. Go to Step \ref{step:gentime}.
            \item If the site is \textit{empty} execute the reaction and update data structures. Go to Step \ref{step:updatetime}.
        \end{itemize}

        \item In case of a contact-promoted reaction randomly select a neighboring target site of the target cell from its adjacent sites, with discrete uniform distribution. 
        \begin{itemize}
            \item If the site is \textit{empty}, reject the reaction. Go to Step \ref{step:gentime}.
            \item If the site is \textit{occupied}, or it is outside the border, execute the reaction, update data structures. Go to Step \ref{step:updatetime}.
        \end{itemize}
	\end{itemize}
	
	\item \textbf{Update time:}\label{step:updatetime} set $t\leftarrow t+\hat{T}$.
	
	\item \textbf{Halt if $H = True$ else continue the process with Step 2.}
\end{enumerate}
}

\subsection{The Reduced Rate Method (RRM)}\label{sec:RRA}

The essence of the method is the classification of the cells according to the number of their free neighbors on the lattice. 
From a theoretical viewpoint, this approach enables us to define a state space in which the system is 'well-stirred' in the sense that the probability of a reaction to happen in the next time interval $\delta t$ does not depend on the spatial localization of the reactants (cells). Thus, with the proper scaling of the reaction rates and some appropriate assumptions we reduce the problem to Gillespie's SSA formalism.

\subsubsection*{Assumptions in the RRM algorithm}

According to the restrictions we made, the products of the cellular and non-cellular (bio)chemical  reactions neither change the state of the cells nor diffuse: the cells only differ in their positions on the lattice.
Based on this assumption, we may classify the cells at state $\underbar{X}(t)$, according to the number of their free adjacent sites $j$. 
We refer to cells that have $j$ free adjacent sites \textit{class-j cells}.
At any state $\underbar{X}(t)$, the value of $N_j=N_j(t)$ for $\forall j\in\{0,\dots,4\}$ and the positions of the cells in class-$j$ are known, thus, we can define the propensities and choose the target cells based on this information.

We assume that \textbf{contact dependent reactions occur between cells and their} ($j$ free or $(4-j)$ occupied) \textbf{neighbors}.
Thus, in the RRM algorithm selecting a reaction defines two things: i) the reaction index $\mu$, that defines the change in the state $\underbar{X}$ and ii) the class-$j$ of the target cell.
We define reaction $R_{\mu,j}$, that is reaction with index $\mu$ acting on a cell having $j$ free neighbors (and possibly causes a $(\Delta S_1,\dots,\Delta S_m)$  change in the number of the chemical species).
\subsubsection*{Propensities and probabilities}

For all $\mu\in\{1,\dots,M\}$ and all $j\in\{0,\dots,4\}$ we define propensity functions $a_{\mu,j}$ according to the following considerations.

\begin{subequations}
\label{eq:rr_prop}
   \textbf{ Contact-inhibited reactions} occur when a cell 'reacts' with one of its empty neighbors. Thus, in this case the reactants are the cells and their empty adjacent sites on the lattice. The number of possible combinations between the class-$j$ cells and their free adjacent sites is $j\cdot N_j$.
    The reaction rate of an elementary reaction, with one empty neighbor, is $r_\mu/4$. 
    Thus, the propensity function of the contact-inhibited reaction $\mu$ acting on a cell with $j$ free adjacent sites is:
\begin{equation}
    \label{eq:rr_prop_alpha_beta}
    a_{\mu,j}=\frac{r_\mu}{4}jN_j, \hspace{.3cm} \forall j\in\{0,\dots,4\}.
\end{equation}

The propensities of \textbf{spontaneous reactions} does not depend on the number of free neighbors of the cell. Nevertheless, we use the two-index notation to simplify the proof later.
\begin{equation}
    \label{eq:rr_prop_delta_omega}
    a_{\mu,j}=r_\mu N_j, \hspace{.3cm} \forall j\in\{0,\dots,4\}.
\end{equation}
For non-cellular chemical reactions we use the same scaling we used in the PDM algorithm:
\begin{equation}
    \label{eq:rr_prop_non-cel}
    a_{\mu, j}=\frac{c_\mu h_\mu}{N} N_j=r_\mu N_j\hspace{.3cm} \forall j\in\{0,\dots,4\}.
\end{equation}

\textbf{Contact-promoted reactions} occur when a cell 'reacts' with one of its occupied neighbors. 
The number of possible combinations between the class-$j$ cells and their occupied $4-j$ neighbors is: $(4-j)\cdot N_j$. The rate of such an elementary reaction is $r_\mu/4$. The propensity function of the contact-promoted reaction $\mu$ acting on a class-$j$ cell is:
\begin{equation}
    \label{eq:rr_prop_kappa}
    a_{\mu,j}=\frac{r_\mu}{4}(4-j) N_j(t),
    \hspace{.3cm} \forall j\in\{0,\dots,4\}.
\end{equation}
\end{subequations}
We assume that reaction indices $\mu\in\{1,\dots,\alpha\}$ belong to the contact-inhibited, $\mu\in\{\alpha+1,\dots,\alpha+\beta\}$ belong to the spontaneous and $\mu\in\{\alpha+\beta+1,\dots,M\}$ belong to the contact dependent reactions.
The sum of the propensities is:
\begin{equation}
\label{eq:rr_a}
\begin{split}
a&= \sum_{\mu=1}^M\sum_{j=0}^4 a_{\mu,j}=
\frac{1}{4}\sum_{\mu=1}^\alpha r_\mu\sum_{j=0}^4jN_j+
\sum_{\mu=\alpha+1}^{\alpha+\beta} r_\mu N+
\frac{1}{4}\sum_{\mu=\alpha+\beta+1}^{M} r_\mu \sum_{j=0}^4(4-j)N_j
\end{split}
\end{equation}
Unlike in the PDM, in case of the RRM $a$ may be zero in some models even with $N>0$.
We have to ensure that this does not cause problems while drawing the waiting time.

With this notations $a_{\mu,j}\delta t$ is the probability, to first order in $\delta t$, that a reaction with index $\mu$ will occur to a class-$j$ cell on the lattice and possibly cause a $(\Delta S_1,\dots,\Delta S_m)$ change in the number of the chemical species in the next time interval $\delta t$.
We denote the corresponding event $R_{\mu,j}$. 

The propensity functions explicitly specify the corresponding $(\mu,j)$ pairs, thus, we obtain the joint probability of executing reaction $\mu$ on a class-$j$ cell.
With these considerations and notations, with $a>0$ the probability density function of the intervent time $\tau$ until the next reaction is $P_1(\tau)$ and the joint probability that the next reaction would be of index $\mu$ and would be executed on a class-$j$ cell given waiting time $\tau$ is $P_2((\mu,j)|\tau)$:
\begin{equation}
    \label{eq:rr_pmf}
    P_1(\tau)=a e^{-a\tau},
    \hspace{0.3cm}
    P_2\left((\mu,j)|\tau\right)=\frac{a_{\mu,j}}{a}=\vcentcolon P(R_{\mu,j}), \hspace{0.3cm} \mu\in\{1,\dots,M\},\hspace{0.2cm} j\in\{0,\dots,4\}
\end{equation}

\subsubsection*{Data structure}
Note that in case of a movement, proliferation, or a cell death event, the class of several cells can change at once. Thus, we need to reach the cells as quick as possible to adjust the corresponding changes. 
A convenient data structure for this (in Python) consists of five dictionaries, that we will store in one list, called $\underbar{D}$. Dictionary $\underbar{D}[j]$ contains the cells that have $j$ free adjacent sites in the following way: the indices of cells are the keys and the corresponding coordinates are the values. Naturally, the length of $\underbar{D}[j]$ is $N_j$. We also store the amount of chemical substances in a dedicated list $\underbar{B}$.

We keep the underlying lattice $L$ to check if a particular lattice site is occupied or empty. 
A convenient data structure is a matrix $L\in \mathbb{N}^{n_1\times n_2}$, where entry $L[x_*,y_*]$ is the index of the cell in the population. 
Note that every newborn cell has to be given a unique index to avoid complications. 
Thus, in some models the value of some indices may be greater than the capacity $K$ of the population. With these notations and assumptions, the Reduced Rate Method is as follows.

{\smaller
\paragraph{Algorithm Reduced Rate Method}
\begin{enumerate}
	\item \textbf{Initialization:}
	\begin{alphlist}
	\item set $t\leftarrow 0$, prescribe halting conditions $H$,
    \item initialize the lattice $L$ with size $n_1\times n_2$, initialize $\underbar{D}$,\\
    place $N(0)\leq n_1\cdot n_2$ cells on the lattice according to essay,\\ store location of class-$j$ cells in dictionary $\underbar{D}[j]$,
    \item initialize $\underbar{B}$,
    store initial amount of chemical species in list $\underbar{B}$.
	\end{alphlist}
	\item \textbf{Calculate propensity functions} $a_{\mu,j}$ and $a$ for all $\mu\in\{1,\dots,M\}$ and $j\in\{0,\dots,4\}$ according to Eqs.
	(\ref{eq:rr_prop}) and (\ref{eq:rr_a}). If $a=0$ halt the algorithm.
	\item \textbf{Generate intervent time:} choose $\tau$ according to $P_1(\tau)$ from Eq. (\ref{eq:rr_pmf}).
	\item \textbf{Select event and cell type:} choose $(\mu,j)$ according to $P(R_{\mu,j})$ from Eq. (\ref{eq:rr_pmf}). \\
	Select target cell from class $j$ with discrete uniform distribution. \\
	In case of a movement, or cell division, select one of the $j$ free adjacent sites of the target cell with discrete uniform distribution.
	\item \textbf{Execute and update:}
	\begin{alphlist}
	    \item Update time: set $t\leftarrow t+\tau$.
	    \item Execute the event.
		\item Update $\underbar{D},\underbar{B}$ and $L$.
	\end{alphlist}
	\item \textbf{Halt if $H = True$ else continue the process with Step (2).}
\end{enumerate}
}

\subsection{The marginal RRM algorithm (mRRM)}

Next, we describe a variant of the RRM in which the reaction index $\mu$ and the cell type $j$ are selected in two distinct steps.
With this approach, we may reduce the number of propensity functions which may make the implementation more convenient and may reduce the simulation runtime. 
In the mRRM algorithm we use the same data structure we introduced in case of the RRM algorithm. 
Assume that reactions with index $\mu\in\{1,\dots,\alpha\}$ are contact-inhibited, $\mu\in\{\alpha+1,\dots,\alpha+\beta\}$ are spontaneous and $\mu\in\{\alpha+\beta+1,\dots,M\}$ are contact dependent reactions.

\subsubsection*{Propensities and probabilities}
Consider the propensity functions $a_{\mu,j}$ of the RRM (Eq. \ref{eq:rr_prop}) and the joint probabilities $P(R_{\mu,j})=a_{\mu,j}/a$ obtained from them.
\noindent First, for every reaction $\mu\in\{1,\dots,M\}$, we sum over $j$ obtaining the marginal probabilities $P(R_{\mu})=a_{\mu}/a$ for reaction $R_\mu$. The corresponding propensities $a_\mu=\sum_{j=0}^4a_{\mu,j}$ are of form:
\begin{equation}
\label{eq:RRAm}
a_{\mu} = \left\{
        \begin{array}{ll}
        \frac{r_\mu}{4}\sum_{j=0}^4jN_j(t) & \quad \mu\in\{1,\dots,\alpha\} \\
        r_\mu N(t)& \quad \mu\in\{\alpha+1,\dots,\alpha+\beta\} \\
        \frac{r_\mu}{4}\sum_{j=0}^4(4-j) N_j(t) & \quad \mu\in\{\alpha+\beta+1,\dots,M\}
    \end{array}
    \right.
\end{equation}
\noindent Naturally, the sum of the propensity functions are the same in the RRM and mRRM algorithms:
\begin{equation}\label{mRRM_a}
    a=\sum_{\mu=1}^M\sum_{j=0}^4a_{\mu,j}=\sum_{\mu=1}^M a_\mu.
\end{equation}

\noindent In case $a>0$, the probability density function $P_1(\tau)$ of the waiting time $\tau$ until the next reaction, and the probability $P_2(\mu|\tau)$ of the next reaction $R_\mu$ are:
\begin{equation}
    \label{eq:RRm_pmf}
    P_1(\tau)=a e^{-a\tau},
    \hspace{0.5cm}
    P_2\left(\mu|\tau\right)=\frac{a_{\mu}}{a}=\vcentcolon P(R_{\mu}), \hspace{0.5cm} \mu\in\{1,\dots,M\}.
\end{equation}

\subsubsection*{Selecting the target cell}

Finally, we have to obtain a method for selecting the target cell in the mRRM algorithm. 
It would be tempting to simply choose the class-$j$ of the target cell according to $N_j/N$, in case of all reaction types.
But this would lead to a wrong result in general. 
Suppose, we chose cell movement to execute and suppose $N_0>0$. 
There would be a positive probability to choose class-0 from which to select the target cell, and thus select a cell with no free neighbors, even though such a cell is not able to move.

In the original RRM algorithm the probability of executing reaction $\mu$ on a type $j$ cell is $P(R_{\mu,j})$. 
We saw that in the mRRM algorithm, selecting reaction $R_\mu$ and selecting the type $j$ of the target cell is executed in two distinct steps, but these steps are not independent from each other. 
In the mRRM algorithm, let $A_j$ denote the event that the target cell is of class $j$.
For every reaction $R_\mu$, we look for the probability $P(A_j|R_\mu)$, with which $P(A_j|R_\mu)\cdot P(R_\mu)=P(R_{\mu,j})$ in the original RRM algorithm. That is:

\begin{equation}
\label{eq:P(Aj|Rmu)}
P(A_j|R_\mu) = \left\{
        \begin{array}{ll}
        \frac{jN_j}{\sum_{j=0}^4jN_j} & \quad \mu\in\{1,\dots,\alpha\} \\
        N_j/N & \quad \mu\in\{\alpha+1,\dots,\alpha+\beta\} \\
        \frac{(4-j)N_j}{\sum_{j=0}^4(4-j)N_j} & \quad \mu\in\{\alpha+\beta+1,\dots,M\}
    \end{array}
    \right.
\end{equation}

Notice that with Eq. (\ref{eq:P(Aj|Rmu)}) it is not possible to select class-$0$ in case of a contact-inhibited reaction. Likewise, it is not possible to select class-$4$ in case of a contact-promoted reaction.
On Fig. \ref{fig:distributions_mRRM} we illustrate Eq. (\ref{eq:P(Aj|Rmu)}) in a concrete example, the case of the snapshot of the simulation introduced on Fig.\ref{fig:snapshot} (Sec.\ref{sec:lattice_state}).
After selecting the class $j$ of the target cell according to Eq. (\ref{eq:P(Aj|Rmu)}), the algorithm chooses a target cell from the class $j$ cells with discrete uniform distribution.

It easy to see that the RRM and mRRM algorithms are mathematically equivalent formulations to obtain realizations of a given stochastic process concerning a cell culture.

\begin{figure}[h!]
    \centering
    \includegraphics[width=1\textwidth]{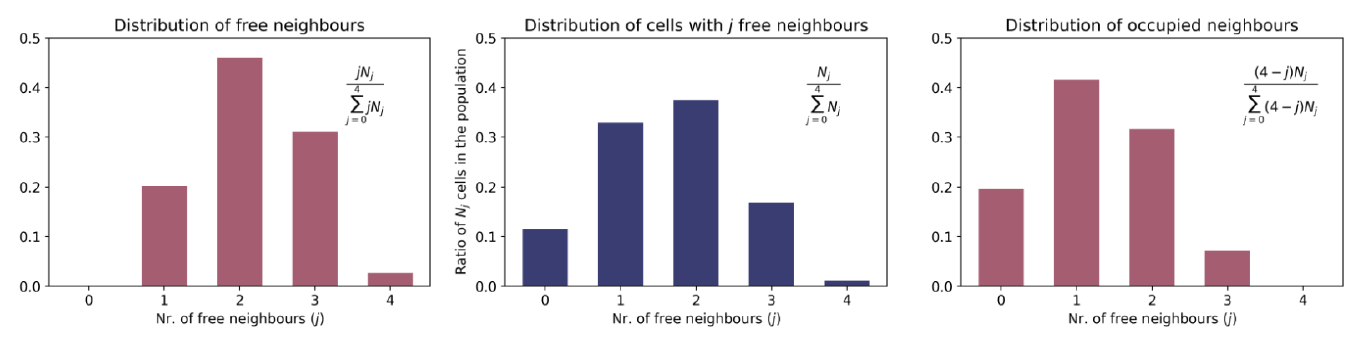}
    \caption{\textbf{Illustration of the distributions in Eq. (\ref{eq:P(Aj|Rmu)})} 
    The figures correspond to the snapshot of Fig.\ref{fig:snapshot}. 
    Suppose we had all three reaction types in the simulation at state $\underbar{X}(t*)$ of the system. 
    In case the next reaction would be a \textit{spontaneous reaction}, we should choose the class $j$ of the target cell according to the distribution on the middle figure, and then choose a particular target cell from class $j$ at random. This is equivalent of choosing a cell randomly from the total population.
    It is clear from the figure that we would choose the cell class $j=2$ cell with the highest probability, as the ratio of cells with two free adjacent sites is the highest in the population. 
    Since there are cells in all five classes, we may choose any class with nonzero probability in a spontaneous reaction.
    If the next reaction would be a \textit{contact-inhibited reaction}, we should choose the class of the target cell according to the distribution on the left figure. Since $0\cdot N_0=0$, we choose class $j=0$ with zero probability. Thus cells with zero free neighbors do not participate in contact-inhibited reactions. At this state of the system we would choose a class $j=2$ cell with the highest probability. 
    If the next reaction would be a \textit{contact-promoted reaction}, we should choose the target cell type from the distribution on the right. Since $(4-4)\cdot N_4=0$, we choose class $j=4$ with zero probability -- cells with 4 free adjacent sites do not participate in contact-promoted reactions. We would choose class $j=1$ with the highest probability at this state.}
    \label{fig:distributions_mRRM}
\end{figure}

{\smaller
\paragraph{Algorithm marginal Reduced Rate Method}
\begin{enumerate}
\item \textbf{Initialization:}
		\begin{alphlist}
	\item set $t\leftarrow 0$, prescribe halting conditions $H$,
    \item initialize the lattice $L$ with size $n_1\times n_2$, initialize $\underbar{D}$,\\
    place $N(0)\leq n_1\cdot n_2$ cells on the lattice according to essay,\\ store location of class-$j$ cells in dictionary $\underbar{D}[j]$,
    \item initialize $\underbar{B}$,
    store initial amount of chemical species in list $\underbar{B}$.
	\end{alphlist}
	\item \textbf{Calculate propensity functions} $a_{\mu}$ and $a$ for all $\mu\in\{1,\dots,M\}$ according to Eq. (\ref{eq:RRAm}) and (\ref{mRRM_a}). \\
	If $a=0$ halt the algorithm.
	\item \textbf{Generate intervent time:} choose $\tau$ according to $P_1(\tau)$ from Eq. (\ref{eq:RRm_pmf}).
	\item \textbf{Select event:} choose $\mu$ according to $P_2(\mu|\tau)$ from Eq. (\ref{eq:RRm_pmf}).
	\item \textbf{Select cell:} 
	\begin{itemize}
	    \item Choose the class $j$ of cells according to $P(A_j|R_\mu)$ from Eq. (\ref{eq:P(Aj|Rmu)}).
	    \item Choose a target cell randomly with uniform distribution from class $j$.
	    \item If $R_\mu$ correspond to movement or division, then choose one of the free neighbors of the target cell with discrete uniform distribution.
	\end{itemize} 
	\item \textbf{Execute and update:}
	\begin{alphlist}
	    \item Update time: set $t\leftarrow t+\tau$.
	    \item Execute the event.
		\item Update $\underbar{D},\underbar{B}$ and $L$.
	\end{alphlist}
	\item \textbf{Halt if $H = True$ else continue the process with Step (2).}
\end{enumerate}
}

\section{Equivalence and comparison of the algorithms}\label{equi}
Before stating our main theorem, we would like to emphasize the differences between the PDM and the RRM algorithms.
The first important difference is in the definition of the propensity functions: in case of the PDM, the propensity functions always depend on the total number of cells (Eq. \ref{eq:prop_ex}). In contrast, in the RRM we classify the cells according to the number of their free neighbors. Only 'spontaneous' reactions and the corresponding propensities depend on the total number of cells (Eq. \ref{eq:rr_prop_delta_omega}) other reactions occur between cells and their empty neighbors (Eq. \ref{eq:rr_prop_alpha_beta}) or between cells and their occupied neighbors (Eq. \ref{eq:rr_prop_kappa}).

The second important difference is that the drawn event does not always occur in the PDA approach, whereas the drawn event is always executed in the RRA approach.
The third important difference is that drawing reaction $\hat{R}_\mu$ and randomly selecting a target cell (that has $j$ free neighbors, event $\hat{E}_j$) are independent in the case of PDM. In the RRM, $R_{\mu,j}$ uniquely determines that the reaction with index $\mu$ has to be executed on a cell having $j$ free neighbors.

During the proof we assume that reactions with index $\mu\in\{1,\dots,\alpha\}$ are contact-inhibited, $\mu\in\{\alpha+1,\dots,\alpha+\beta\}$ are spontaneous and $\mu\in\{\alpha+\beta+1,\dots,M\}$ are contact dependent reactions. We also assume that their order of reaction indices are the same in both the PDM and RRM algorithms.

\begin{theorem}[On the equivalence of the PDM and RRM algorithms]
\label{thm:stateq}
Using the notations introduced so far, consider the state $\underbar{X}(t)$ of the system at time $t$. Suppose that, in this state both $\hat{a},a>0$ in PDM and RRM, respectively. Suppose we start both algorithms from $\underbar{X}(t)$ (i. e., $\hat{T}=0$ in case of the PDM). 

In the PDM algorithm let $\hat{R}_\mu$ denote the drawn reaction in the $k$-th attempt during the 'Reaction selection step' and let $\hat{E}_j$ denote the event that we choose a target cell with $j$ free neighbors in the corresponding 'Cell selection step' and let $\hat{S}$ denote the random event that we would be able to execute $\hat{R}_\mu$ in the 'Decision step'.

In the RRM algorithm let $R_{\mu,j}$ denote the next reaction. 
Let $\hat{T}$ and $\tau$ denote the soujurn times until the next executed reaction in the PDM and the RRM, respectively. Then
\begin{alphlist}
    \item $\hat{P}((\hat{R}_\mu\cap\hat{E}_j)|\hat{S})=P(R_{\mu,j})$, and
    \item $\hat{T}\sim Exp(a)$ and $\tau\sim Exp(a)$ with parameter $a$ defined in (\ref{eq:rr_a}).
\end{alphlist}
\end{theorem}
\begin{proof} We will use a direct proof. Let us start with the probability distribution of events.

\textit{(a)} In the RRM algorithm $P(R_{\mu,j})=a_{\mu,j}/a$ (Eq. \ref{eq:rr_pmf}). We will show that in the PDM $\hat{P}((\hat{R}_\mu\cap\hat{E}_j)|S)=a_{\mu,j}/a$ with the definition of $a_{\mu,j}$ (Eq. \ref{eq:rr_prop}) and $a$ (Eq. \ref{eq:rr_a}).

In the PDM algorithm let $\hat{P}(\hat{S})$ denote the probability of the random event $\hat{S}$, that is, we would execute the selected $\hat{R}_\mu$ in the 'Decision step'. 
With Bayes theorem we obtain:
\begin{equation}
    \label{eq:ev_rr}
    \hat{P}(\hat{R}_\mu\cap \hat{E}_j|\hat{S})=\frac{\hat{P}(\hat{S}|\hat{R}_\mu\cap \hat{E}_j)\cdot \hat{P}(\hat{R}_\mu\cap \hat{E}_j)}{\hat{P}(\hat{S})}.
\end{equation}

During the 'Cell selection step' in the PDA algorithm, $\hat{E}_j$ is the event when we choose a target cell that has $j$ free neighbors.
The probability of this is the proportion of cells with $j$ neighbors in the population: $\hat{P}(\hat{E}_j)=N_j/N$. Now, because $\hat{R}_\mu$ and $\hat{E}_j$ are independent events:
\begin{equation}
\label{eq:P(mu,j)}
    \hat{P}(\hat{R}_\mu\cap \hat{E}_j)=\hat{P}(\hat{R}_\mu)\cdot \hat{P}(\hat{E}_j) = \hat{a}_\mu/\hat{a} \cdot N_j/N.
\end{equation}

In the PDM, during the \textbf{Decision step} spontaneous reactions are executed with a probability of 1. Contact-inhibited reactions are executed if we choose a target cell that has at least one free neighbor and we choose one of its free neighbors. Thus, in the PDM algorithm, we introduce the random event $\hat{F}$, which occurs when a contact-inhibited reaction is executed. contact-promoted reactions are executed if we choose a target cell that has occupied neighbor and we choose one of them. We shall denote the corresponding event $\hat{\overline{F}}$. Since, $\hat{R}_\mu$ and $\hat{F}$ are independent, we may obtain:
\begin{equation}
\label{eq:Proof_P(S)}
    \hat{P}(\hat{S})=
    \frac{\sum_{\mu=1}^\alpha\hat{a}_\mu}{\hat{a}}\hat{P}(\hat{F})+
    \frac{\sum_{\mu=\alpha+1}^{\alpha+\beta}\hat{a}_\mu}{\hat{a}}+
    \frac{\sum_{\mu=\alpha+\beta+1}^{M}\hat{a}_\mu}{\hat{a}}\hat{P}(\hat{\overline{F}}).
\end{equation}

Now, we need the probabilities $\hat{P}(\hat{F})$ and $\hat{P}(\hat{\overline{F}})$. The probability that a randomly chosen neighbor of a cell is free, given that the cell has $j\in\{0,\dots,4\}$ free neighbors is $\hat{P}(\hat{F}|\hat{E}_j)=j/4$. Since $\hat{E}_0,\dots,\hat{E}_4$ are pairwise disjoint events whose union is the entire sample space (i. e., we randomly choose a cell from the population that has 0 to 4 free von Neumann neighbors), we may obtain $\hat{P}(\hat{F})$ by the law of total probability:
\begin{equation}
    \hat{P}(\hat{F})=\sum_{j=0}^4\left(\hat{P}(\hat{F}|\hat{E}_j)\cdot \hat{P}(\hat{E}_j)\right)=\frac{1}{4N}\sum_{j=0}^4jN_j.
\end{equation}
Since $\hat{\overline{F}}$ is the complement event of $\hat{F}$: $\hat{P}(\hat{\overline{F}})=1-\hat{P}(\hat{F})=\frac{1}{4N}\sum_{j=0}^4(4-j)N_j$. Thus, we have a formula for every term in Eq. (\ref{eq:Proof_P(S)}).

Finally, we have to evaluate $\hat{P}(\hat{S}|\hat{R}_\mu\cap \hat{E}_j)$ that is the probability that the currently selected reaction can be executed given that it is a reaction $\hat{R}_\mu$ and the target cell has $j$ free neighbors. Given the reaction and the target cell, we have to choose one neighbor of the target cell at random with discrete uniform distribution. In case of a CIR we can execute $\hat{R}_\mu$ if the site is empty, with probability $j/4$. Similarly, a CPR can be executed if the chosen site is occupied, it occurs with probability $(4-j)/4$, and we can always execute a spontaneous reaction.
For all pairs $(\mu,j)$, where $j\in\{0,\dots,4\}$:

\begin{equation}
\label{eq:P(S|mu,j)}
\hat{P}(\hat{S}|\hat{R}_\mu\cap \hat{E}_j) = \left\{
        \begin{array}{ll}
            j/4 & \quad \mu\in\{1,\dots,\alpha\} \\
            1 & \quad \mu\in\{\alpha+1,\dots,\alpha+\beta\} \\
            (4-j)/4 & \quad \mu\in\{\alpha+\beta+1,\dots,M\}
        \end{array}
    \right.
\end{equation}

\noindent Now, from Eqs. (\ref{eq:prop_ex}), (\ref{eq:rr_a}),  (\ref{eq:Proof_P(S)}) and from the definition of $\hat{a}$ we have:
\begin{equation}\label{eq:Proof_aS}
    \hat{a}\hat{P}(\hat{S})=\frac{1}{4}\sum_{\nu=1}^\alpha r_\nu\sum_{j=0}^4jN_j+\sum_{\nu=\alpha+1}^{\alpha+\beta}r_\nu N+\frac{1}{4}\sum_{\nu=\alpha+\beta+1}^{M}r_\nu\sum_{j=0}^4(4-j)N_j=a.
\end{equation}

\noindent Thus, in the case of the contact-inhibited reaction $\mu\in\{1,\dots,\alpha\}$, for the probability (\ref{eq:ev_rr}), with using (\ref{eq:P(mu,j)}), (\ref{eq:Proof_P(S)}),  (\ref{eq:P(S|mu,j)}), and \ref{eq:Proof_aS} we obtain:
\begin{equation}
\begin{split}
    \hat{P}(\hat{R}_{\mu}\cap \hat{E}_{j}|\hat{S})
    &=\frac{\frac{j}{4}\frac{N_j}{N}\cdot\frac{\hat{a}_\mu}{\hat{a}}}{
    \hat{P}(\hat{S})}=\frac{\frac{1}{4N}(jN_j)\cdot r_\mu N}
    {a}=\frac{\frac{r_\mu}{4}jN_j}{a}=\frac{a_{\mu,j}}{a}=P(R_{\mu,j})
\end{split}
\end{equation}
For the spontaneous reactions $\mu\in\{\alpha+1,\dots,\alpha+\beta\}$:
\begin{equation}
\begin{split}
    \hat{P}(\hat{R}_{\mu}\cap \hat{E}_{j}|\hat{S})
    &=\frac{1\cdot\frac{N_j}{N}\cdot\frac{\hat{a}_\mu}{\hat{a}}}{\hat{P}(\hat{S})}=\dots=\frac{a_{\mu,j}}{a}=P(R_{\mu,j})
\end{split}
\end{equation}
For the contact-promoted reactions $\mu\in\{\alpha+\beta+1,\dots,\alpha+\beta+\gamma\}$:
\begin{equation}
\begin{split}
    \hat{P}(\hat{R}_{\mu}\cap \hat{E}_{j}|\hat{S})
    &=\frac{\frac{4-j}{4}\cdot\frac{N_j}{N}\cdot\frac{\hat{a}_\mu}{\hat{a}}}{\hat{P}(\hat{S})}=\dots=\frac{a_{\mu,j}}{a}=P(R_{\mu,j})
\end{split}
\end{equation}
This proves part \textit{(a)} of the theorem.

\textit{(b)} We now show that in PDM the sojourn time $\hat{T}$ until the next successful event has the same probability density function as the sojourn time $\tau$ in RRM.
From the previous part of the proof we know that the reaction in PDM is only successful with probability $\hat{P}(\hat{S})$ (Eq.\ref{eq:Proof_P(S)}), thus we discard the reaction with probability $1-\hat{P}(\hat{S})$. Let us introduce the discrete random variable $\xi\in\{1,2,\dots\}$ that, in a given state $\underbar{X}(t)$ of the system, counts how many times the PDM algorithm has to repeat the Decision step in one iteration, until a reaction can be executed. In this case, the reaction is discarded $\xi-1$ times until the next realized events.

$\xi$ is of a geometric distribution with parameter $\hat{P}(\hat{S})$, the probability of repeating the Decision step $k$ times in one iteration of the algorithm is:
{\small
\begin{equation}
\label{eq:P(xi=k)}
    \hat{P}(\xi=k)=(1-\hat{P}(\hat{S}))^{k-1}\hat{P}(\hat{S})
\end{equation}
}
In every trial $l\in\{1,\dots,k\}$ the algorithm generates a random sojourn time $\tau_l\sim Exp(\hat{a})$. 
Given that there were $k$ trials, and since $\hat{\tau}_1,\dots \hat{\tau}_k$ are independent and exponentially distributed with parameter $\hat{a}$, the sum of these soujurn times $\hat{T}=\sum_{l=1}^k\hat{\tau}_l$ is a continuous random variable of Erlang-$k$ distrubution. Its probability density function can be calculated with convolution resulting $f(t,k,\hat{a})$ and its cumulative distribution function is $F(t,k,\hat{a})$ with $ t, \hat{a}>0,$ and $ k\in\{1,2,\dots\}$ \cite{StatDist}:
{\small
\begin{equation}
\nonumber
    f(t,k,\hat{a})=\frac{\hat{a}^k\cdot  t^{k-1}}{(k-1)!}\cdot e^{-\hat{a}t} \hspace{1cm}
    F(t,k,\hat{a})=
    1-e^{-\hat{a}t}\sum_{l=0}^{k-1}\frac{1}{l!}(\hat{a}t)^l.
\end{equation}
}
\noindent We may obtain a formula for the joint probability $\hat{P}(\hat{T}<t,\xi)$ by conditioning on $\xi$:
{\small
\begin{equation}\nonumber
  \hat{P}\left(\sum_{l=1}^k\hat{\tau}_l<t,\xi\!=\!k\right)=
  \hat{P}\left(\sum_{l=1}^k\hat{\tau}_l<t|\xi\!=\!k\right)\cdot\hat{P}(\xi\!=\!k)=
  \left(1-e^{-\hat{a}t}\sum_{l=0}^{k-1}\frac{1}{l!}(\hat{a}t)^l\right)
  \cdot(1-\hat{P}(\hat{S}))^{k-1}\hat{P}(\hat{S}).
\end{equation}
}
The marginal distribution $\hat{P}(\hat{T}<t)$ can be derived with summation for $k$:
{\small
\begin{align}
\nonumber
  \hat{P}(\hat{T}<t)
  &=\sum_{k=1}^{\infty}\left(1-e^{-\hat{a}t}\sum_{l=0}^{k-1}\frac{1}{l!}(\hat{a}t)^l\right)
  \cdot(1-\hat{P}(\hat{S}))^{k-1}\hat{P}(\hat{S})\\\nonumber
  &=\sum_{k=1}^{\infty}(1-\hat{P}(\hat{S}))^{k-1}\hat{P}(\hat{S})-\sum_{k=1}^{\infty}e^{-\hat{a}t}\sum_{l=0}^{k-1}\frac{1}{l!}(\hat{a}t)^l(1-\hat{P}(\hat{S}))^{k-1}\hat{P}(\hat{S})
  .
\end{align}}
Where the first sum is trivially 1. To obtain the second sum, we rearrange the terms and change the summation index $l$ to $m=l+1$:
{\small
\begin{align}
\nonumber
  \hat{P}(\hat{T}<t)
  &=1-e^{-\hat{a}t}\sum_{k=1}^{\infty}\sum_{m=1}^{k}\frac{(\hat{a}t)^{m-1}}{(m-1)!}(1-\hat{P}(\hat{S}))^{k-1}\hat{P}(\hat{S})
  \shorteqnote{change summation order}
  \\\nonumber
  &=1-\hat{P}(\hat{S})e^{-\hat{a}t}\sum_{m=1}^{\infty}\frac{(\hat{a}t)^{m-1}}{(m-1)!}\sum_{k=m}^{\infty}(1-\hat{P}(\hat{S}))^{k-1} \shorteqnote{let $n=k-m$}\\\nonumber
  &=1-\hat{P}(\hat{S})e^{-\hat{a}t}\sum_{m=1}^{\infty}\frac{(\hat{a}t)^{m-1}}{(m-1)!}\sum_{n=0}^{\infty}(1-\hat{P}(\hat{S}))^{n+m-1} \\\nonumber
  &=1-\hat{P}(\hat{S})e^{-\hat{a}t}\sum_{m=1}^{\infty}\frac{(\hat{a}t(1-\hat{P}(\hat{S})))}{(m-1)!}^{m-1}\sum_{n=0}^{\infty}(1-\hat{P}(\hat{S}))^{n}
  \\\nonumber
  &=1-e^{-\hat{a}t}e^{\hat{a}t(1-\hat{P}(\hat{S}))}=1-e^{-\hat{a}\hat{P}(\hat{S})t},
\end{align}}
\noindent since $\sum_{n=0}^{\infty}(1-\hat{P}(\hat{S}))^{n}=1/\hat{P}(\hat{S})$.
With Eq. (\ref{eq:Proof_aS}), we have $\hat{P}(\hat{T}<t)=1-e^{-at}$, which is the cumulative distribution function of the exponential distribution with parameter $a$. Hence, we proved that $\hat{T}\sim Exp(a)$ with parameter $a$ defined in the RRM algorithm (Eq. \ref{eq:rr_a}).
\end{proof}

\begin{corollary}[On the interchangeability of the algorithms]\label{cor:change}
The immediate consequence of the theorem is that in any state of the system, with $\hat{a}, a>0$, after the last step of the currently running algorithm (which is step 8 for PDM and step 6 for RRM), the running algorithm can be interchanged with the other one.
\end{corollary}

\subsection{Running time}
We may expect that different machine times are required due to the different formulation of the algorithms to execute a full iteration in a given state of the system.
The key point between the PDM and the RRM is that the PDM may repeat Steps 3-6 several times according to the model and the state of the system. We shall examine this problem in this section.

\subsubsection*{The number of trials in the PDM algorithm}

First, we obtain an exact formula for the expected number of trials $m_\xi$ (the number of repeating the Decision step between two consecutive, realized reactions) in the PDM algorithm.

\begin{proposition}\label{proposition1}
In the PDM algorithm, given that the system is at state $\underbar{X}(t)$, during one iteration the expected number of trials is:  

\begin{equation}
\label{eq:mk}
    m_\xi=\frac{\sum_{\mu=1}^Mr_\mu N}{\frac{1}{4}\sum_{\nu=1}^\alpha r_\nu\sum_{j=0}^4jN_j+\sum_{\nu=\alpha+1}^{\alpha+\beta}r_\nu N+\frac{1}{4}\sum_{\nu=\alpha+\beta+1}^{M}r_\nu\sum_{j=0}^4(4-j)N_j}.
\end{equation}
Where $N$ is the number of cells in the population and $N_j$ is the number of cells with $j$ free neighbors in state $\underbar{X}(t)$.
\end{proposition}
\begin{proof}
As shown in the proof of the equivalence, the number of trials during one iteration $\xi$ is of geometric distribution with parameter $\hat{P}(\hat{S})$.
Thus, the expected number of trials is $m_\xi=(\hat{P}(\hat{S}))^{-1}$. After simplification, we arrive at the relationship of the proposition.
\end{proof}

\begin{proposition}\label{prop:2}
In the PDM algorithm, consider a model in which the capacity $K$ of the lattice is finite, i. e., there can be at most $K$ cells in the population.\\ 
If $N\rightarrow K$, then $$m_\xi\rightarrow \frac{\sum_{\mu=1}^M r_\mu}{\sum_{\nu=\alpha+1}^{M}r_\nu}.$$
\end{proposition}
\begin{proof}
Because of $N\rightarrow K$ for the number of empty sites $e=(K-N)$, we have $e\rightarrow 0$. Hence, for the number of cells with 0 free neighbors we have $N_0\rightarrow K$. And because of $N-N_0\rightarrow0$ we have $N-N_0=(N_1+\dots +N_4)\rightarrow0$. And since $N_0,\dots,N_4\geq0$ we have $\sum_{j=0}^4jN_j\rightarrow0K=0$ és $\sum_{j=0}^4(4-j)N_j\rightarrow 4K$.
Taking $N\rightarrow K$ in Eq. (\ref{eq:mk}) and substituting these limits we obtain the proposition.
\end{proof}

\begin{corollary_p}
In the PDM algorithm, consider a model in which the capacity $K$ of the lattice is finite, and suppose that the model only has contact-inhibited reactions ($r_\nu=0$ for $\nu\in\{\alpha+1,\dots,M\}$). Now, if $N\rightarrow K$, then $m_\xi\rightarrow \infty$.
\end{corollary_p}

In this latter case, the number of trials in the PDM algorithm, and also its running time, grows without bonds.
Naturally, this would never happen in practice, as we would stop the simulation whenever $N=K$.
Nevertheless, these propositions show that the number of trials and the simulation time may be extremely large.
In contrast, in the RRM algorithm we never face this issue, since the selected reactions are always executed.

\subsection*{Numerical comparison}
Fig. \ref{fig:runtime} shows a numerical comparison of the running times of the PDM and the mRRM algorithms.
We choose a lattice of size $K=100\cdot100$ and fixed proliferation rate $r_p=1$. Then we experimented with 3 different mobility rates: $r_m=1,0.5,0$. In case of each $r_m$ we randomly occupied $q=N(0)/K$, $q\in\{5\%,10\%,\dots,95\%\}$ of the lattice with initial cells, then run and averaged out 50 simulations.
The markers show the corresponding simulation times $\Theta_{mRRM}(r_m,q)$ for each initial setting $q$ normalized to the running time $\Theta_{PDM}(r_m,q)$ of the PDM algorithm, see the horizontal line on the figure, for easier comparison.

With our implementation and with the data structure we chose, for $r_m=1$ and a small number of initial cells the mRRM is 30\% \textit{slower} than the PDM, however with zero mobility rate the mRRM may be 30\% \textit{faster} than the PDM. In the other hand, for greater than 50\% initial occupancy the mRRM always perform better, in fact it may be up to 10 times faster than the PDM.
Therefore, depending on the implementation and data structure, mRRM can indeed provide a notable speed increase.

In some situations, e. g., in a scratch essay the number of modeled cells may affect the outcome of the simulation as the cells may engage in other reactions, such as the cellular or non-cellular (bio)chemical reactions, that are able to influence the time evolution of the population. 
Therefore, enabling larger lattices can offer extensive benefits in terms of the reliability and the biological relevance of the model.

\begin{figure}[h]
    \centering
    \includegraphics[width=.65\textwidth]{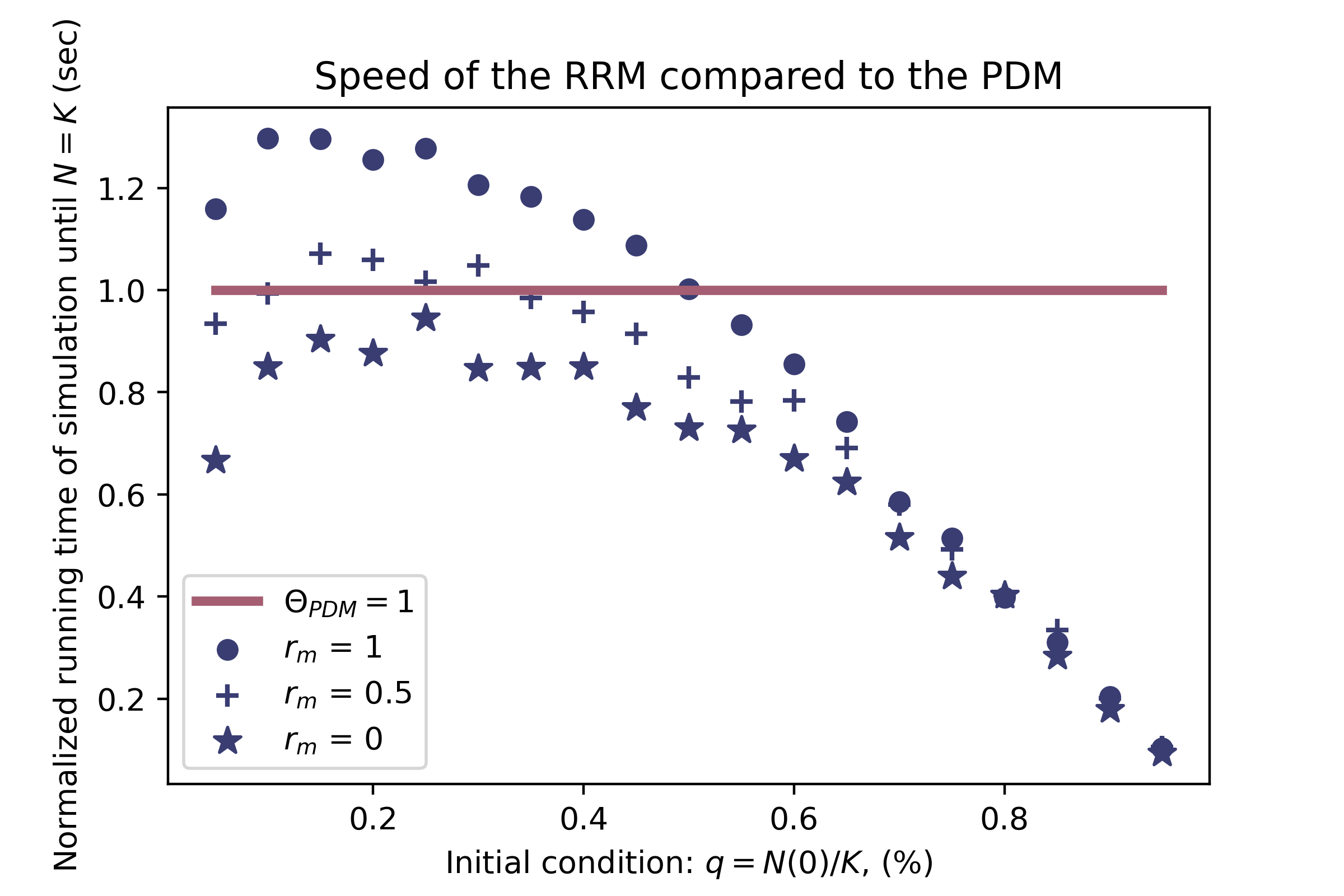}
    \caption{\textbf{Normalized running times of the mRRM algorithm compared to the PDM algorithm.} 
    In these experiments the proliferation rate $r_p=1$ is fixed and different mobility rates $r_m=1,0.5,0$ were set (corresponding to the markers). For each $r_m$ $q\in\{5,10,\dots,95\}$ percent of the lattice was occupied by randomly placed cells, then 50 simulations were run and averaged out.
    We may observe a dramatic relative speed increase of the mRRM algorithm compared to the PDM after 50\% of lattice utilization.}
    \label{fig:runtime}
\end{figure}

\subsection*{Methods to decrease running time}
In Sec. \ref{sec:gil-method} we described several resourceful methods to speed up the algorithm. 
Probably the most convenient approach at first is to only update the propensities if there were changes involving them. 
Rearranging the reaction indices and checking the rapid reactions first is also a feasible approach.
Identifying and separating fast and slow reactions or tau leaping can be easily applied to product synthetic reactions, but must be treated with care in case of motility, division or any type of death event as conflict between reactions may arise (movement of a dead cell to a non empty site, as an extreme example).

Finally, according to corollary \ref{cor:change}, the PDM and any version of the RRM can be used interchangeably. Therefore, after some pre-runs and considering Proposition \ref{proposition1}, one may identify the algorithm that performs the best at the regime of the state space under consideration.

\section{Toy models}\label{sec:toy}
In this section, we present two toy models to illustrate the newly introduced reactions. We used the mRRM algorithm in both simulations. We ran and averaged out 30 stochastic realizations with synthetic parameters. Any of the experiments can be readily repeated with the provided Python code. We present our findings without rigorous mathematical treatment.

\subsection{Contact dependent and spontaneous product synthesis}

Synthesis and secretion of biochemical substances are fundamental physiological processes in living organisms. 
Cells communicate, defend themselves, and influence each other with these processes. 
In this section, we demonstrate experiments with the three essentially different types of synthesis reactions we introduced in Sec. \ref{sec:reactions} (see Table \ref{t:reactions}). As it can be seen on Fig. \ref{fig:prod_syn}, we started to grow a cell population on a lattice with capacity $K=50\cdot50=2500$. 
The initial number of cells were $1\%$ of the capacity, that is $N(0)=25$ cells randomly placed on the lattice. We assumed that the cells engage in contact-inhibited, spontaneous and contact-promoted cellular biochemical synthesis reactions $\mu=2,3,4$, producing one unit amount of chemical substances $S_{ci},S_s,S_{cp}$, respectively. 
We chose $S_{ci}(0)=S_s(0)=S_{cp}(0)=0$ for initial conditions. Let the index of the proliferation reaction be $\mu=1$.
We assumed there are no other reactions in this system. 
We assume that these products are synthesized independently from each other, thus in $r_2,r_3,r_4$ the reactant combinations $h_2=h_3=h_4=1$, we chose $r_1=1$.
The corresponding propensities are:
$$a_1=\frac{r_1}{4}\sum_{j=0}^4jN_j,\hspace{.5cm}
a_2=\frac{r_2}{4}\sum_{j=0}^4jN_j,\hspace{.5cm}a_3=r_3N,\hspace{.5cm} a_4=\frac{r_4}{4}\sum_{j=0}^4(4-j)N_j.$$
With little calculation we can simplify the sum $a$ to: $$a=(r_3+r_4)N+\frac{r_1+r_2-r_4}{4}\sum_{j=0}^4jN_j.$$
The waiting time and the next reaction is selected according to Eq.(\ref{eq:RRm_pmf}) and the class of the proliferating cell is selected according to Eq.(\ref{eq:P(Aj|Rmu)}). 
We made the following assumptions: there is no cell death in the population, the probability of cell division is
$$p_1=\frac{\frac{r_1}{4}\sum_{j=0}^4jN_j(t)}{(r_3+r_4)N+\frac{r_1+r_2-r_4}{4}\sum_{j=0}^4jN_j}>0$$ whenever $0<N<K$, and only zero when $N=K$. Therefore, for the number of cells we expect $N\rightarrow K$ as $t\rightarrow \infty$, as it can be seen on the figure. 

\begin{figure}[h]
    \centering
    \includegraphics[width=.95\textwidth]{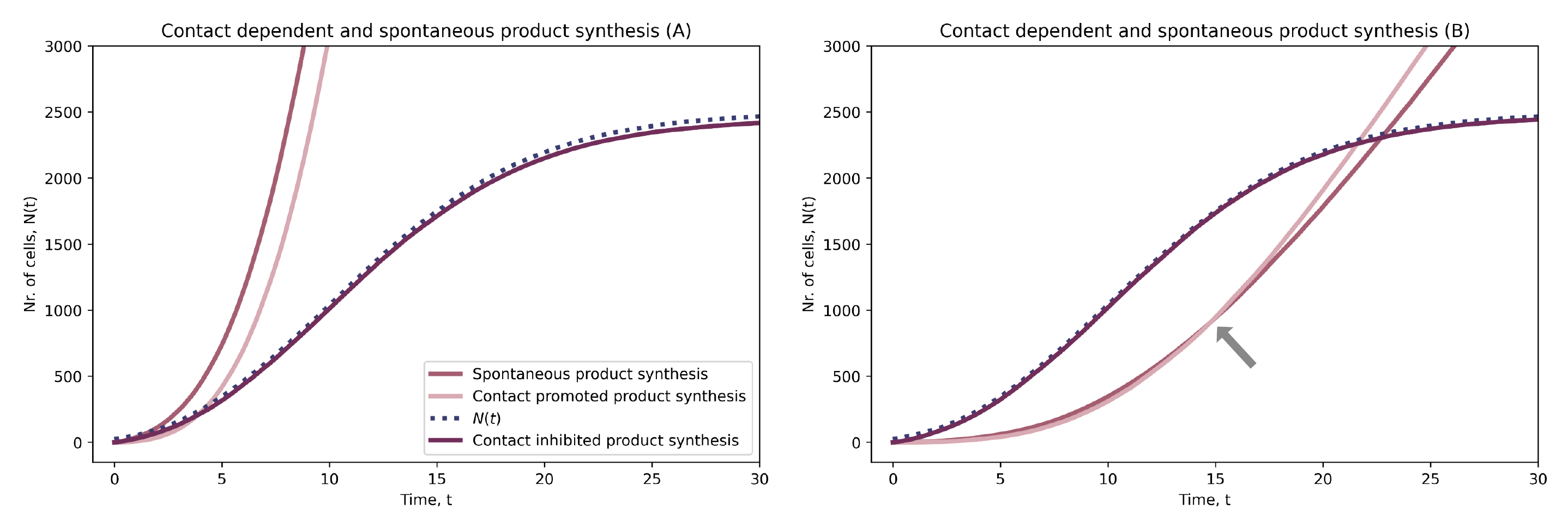}
    \caption{\textbf{The effect of contact dependent and spontaneous product synthesis}}
    \label{fig:prod_syn}
\end{figure}

We now estimate the time evolution of products $S_{ci}(t), S_{s}(t), S_{cp}(t)$.
Following the reasoning of Proposition \ref{prop:2}, it is easy to show that as $N\rightarrow K$, $p_1,p_2\rightarrow0$ and $p_3\rightarrow r_3/(r_3+r_4)$ and $p_4\rightarrow r_4/(r_3+r_4)$. Thus we expect $S_{ci}(t)$ to go into saturation after the size of the population reached the capacity of the lattice, leaving only the spontaneous and contact-promoted synthesis reactions to fire. From their limiting probabilities we may expect that the averaged curves $S_{s}(t), S_{cp}(t)$ approximate straight lines, as the same amount of product $S_s$ and $S_{cp}$ is generated on average per unit time, and the population remains constant after it reached the available capacity. 
Notice that $0<1/4\cdot\sum_{j=0}^4jN_j\leq1/4\cdot\sum_{j=0}^44N_j=N$ for all $N\in\{1,\dots,K\}$.
Thus, in case $r_2\leq r_3$, the probability of the contact-inhibited synthesis reaction is less than or equal to the probability of the spontaneous synthesis reaction. 
We may expect a similar relation between the contact-promoted reaction and the spontaneous reaction as $0\leq1/4\cdot\sum_{j=0}^4 (4-j)N_j=N-1/4\cdot\sum_{j=0}^4 jN_j\leq N$.

In the experiment of subfigure (A) of Fig.\ref{fig:prod_syn}, we chose $r_2=r_3=r_4=1$. We can see that $S_s(t)>S_{cp}(t)$ and the curves smooth to two parallell lines (since $r_2=r_3$) as we may expect based on the above. On subfigure (B), we set $r_2=1, r_3=0.085, r_4=0.1$. 
In the beginning, the amount of product synthesized in the contact-inhibited reaction is much grater than the amount of products generated by the other two reactions, due to the difference in their reaction rates. 
Since $r_3\leq r_4$ $S_s(t)$ and $S_{cp}(t)$ do not align to parallel lines for any $N$, in fact, they intersect as it is indicated in the figure.

\subsection{Contact dependent and spontaneous cell death}

In this section, using this toy model, we illustrate the effect of the three possible types of death that were introduced in Sec.\ref{sec:reactions}. 
Fig. \ref{fig:cell_death} represents the averaged output of the corresponding 30 independent simulations of six different experimental settings.
Specifically, in every single experiment we set the proliferation rate to $r_1=1$ and chose exactly one type of cell death also with rate 1. Initially $N(0)$ cells were randomly placed on the lattice.
The waiting time and the next reaction is selected according to Eq.(\ref{eq:RRm_pmf}) and the class of the target cell is selected according to Eq.(\ref{eq:P(Aj|Rmu)}).

\begin{figure}[h]
    \centering
    \includegraphics[width=.8\textwidth]{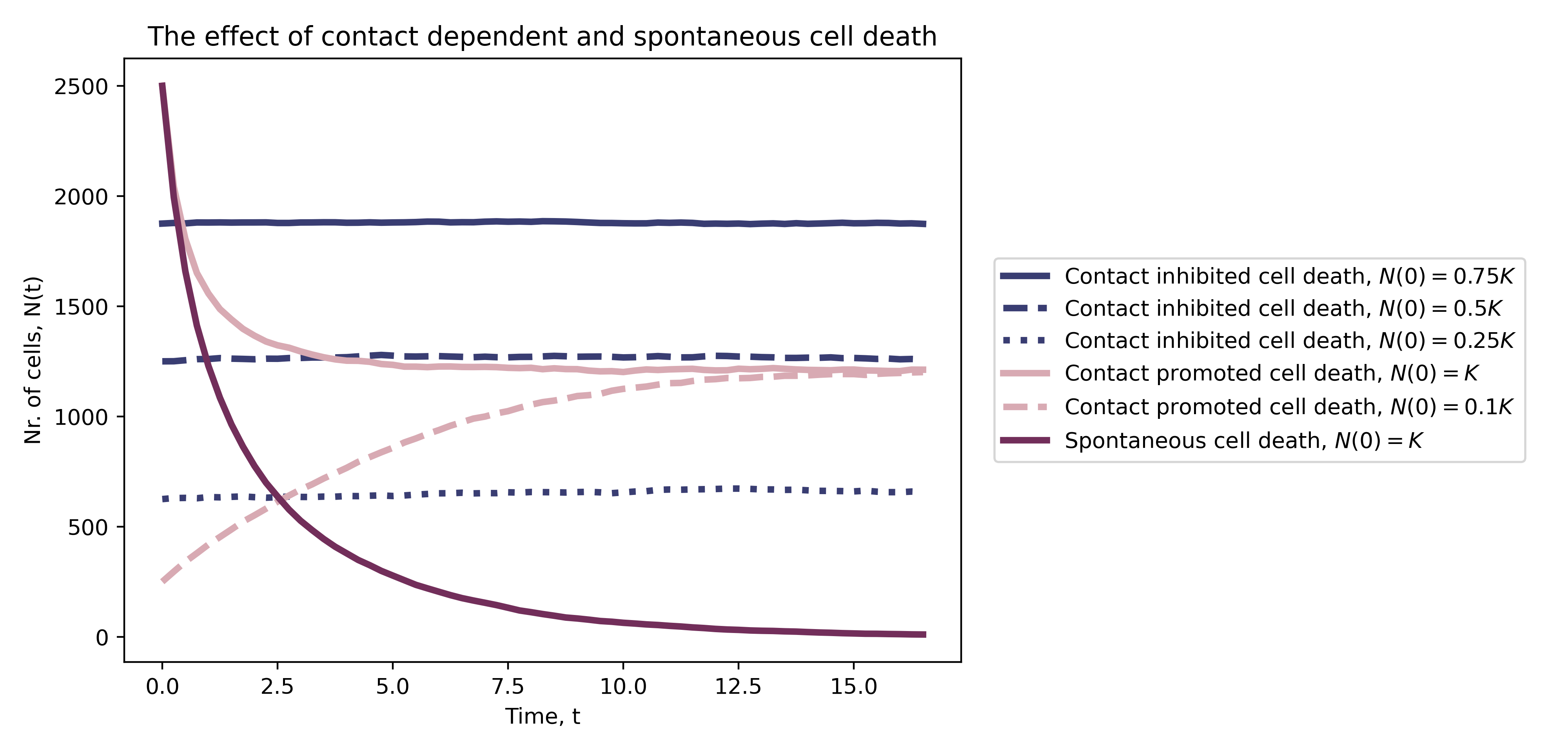}
    \caption{\textbf{The effect of different types of cell death on the population.} The figure shows the average outcome of six independent experiments after 30 runs. 
    We assumed that in all cases the capacity of the lattice was $K=50\cdot50=2500$ and the proliferation rate was $r_1=1$; each curve corresponds to one experiment.  Blue pink, and purple curves illustrate the effect of contact-inhibited, contact-promoted, and spontaneous cell deaths, respectively.}
    \label{fig:cell_death}
\end{figure}

During the first three experiments (blue curves) we assumed contact-inhibited cell death with three different initial conditions, $N(0)=0.75K, 0.5K, 0.25K$. The corresponding propensities and probabilities are:
$$a_{\text{birth}}=\frac{1}{4}\sum_{j=0}^4jN_j(t),
\hspace{.5cm}
a_{\text{death}}=\frac{1}{4}\sum_{j=0}^4jN_j(t),
\hspace{.5cm}
p_{\text{birth}}=p_{\text{death}}=\frac{1}{2}.$$
As proliferation is also a contact-inhibited event, the probabilities of proliferation and death are the same at any state $\underbar{X}(t)$. Thus, we may expect that the population can be sustained in a constant level if $N(0)$ is sufficiently large that $N(t)$ cannot drop to the absorbing state $N=0$ due to stochastic fluctuations.

In our second experiment (pink curves) we assumed contact-promoted cell death and chose $N(0)=K,0.1K$.
The corresponding propensities are:
$$a_{\text{birth}}=\frac{1}{4}\sum_{j=0}^4jN_j(t),
\hspace{.3cm}
a_{\text{death}}=\frac{1}{4}\sum_{j=0}^4(4-j)N_j(t).$$
With both initial conditions the size of the population tends to $K/2>N\approx 1205$ where the probabilities of proliferation and death equals. Notice that these curves are under the blue dashed curve, which is around $K/2$. This is due to the reflecting boundary conditions we use in the experiments (see Sec. \ref{sec:lattice_state}): any cell on the boundary can have at most 3 free adjacent sites.

In the third experiment (purple curve) we assumed spontaneous cell death. The propensities are $$a_{\text{birth}}=1/4\cdot\sum_{j=0}^4jN_j, \hspace{.5cm} a_{\text{death}}=N.$$ And since $r_1=r_2=1$ and $1/4\cdot\sum_{j=0}^4jN_j\leq N$ for $N\in\{1,\dots,K\}$ the probability of death is always greater than the probability of cell division, thus we may expect the population to die out in the long term, as it can be seen on the figure.

\section{Discussion}\label{dis}
In this article, we provided an extended version of the cell simulation algorithm introduced by Baker et al. \cite{CorrBDM} that includes new reaction types. We call this algorithm the Prompt Decision Method (PDM).

In this algorithm we generalized the idea of volume excluding reactions, cell mobility and proliferation, and defined them as reactions that can also take cell-to-cell contact into account: the reactions can be either promoted or inhibited by cellular contact. 
This is reflected by an increased or decreased rate of the corresponding reactions.
We also included spontaneous reactions that are independent from cellular contact.
Possible reaction types include but are not limited to: movement; proliferation; contact-dependent or spontaneous cell death and biochemical reactions.
All of these reactions can be defined in a custom, problem-specific manner, and the number of reactions accounted for by the model is virtually unlimited. 
However, if we limit this model to non-cellular biochemical reactions, it still reduces to Gillespie's chemical stochastic simulation algorithm in a well-stirred environment.
The PDM algorithm is a classical accept–reject montecarlo simulation algorithm that may be very computationally intensive due to the possible large number of rejections in a crowded environment.

We introduced the Reduced Rate Method (RRM) algorithm that is based on the classification of cells according to their free adjacent sites. In this algorithm, every selected reaction is executed, and therefore, the running time and the resource requirements see dramatic reductions.
Furthermore, we also showed that the PDM and the RRM methods generate the same stochastic process as the distribution of the waiting time until the next successful reaction; also, the probabilities of the next successful reaction are the same in both approaches in a given state of the system.
The RRM algorithm is an exact stochastic simulation algorithm that can be derived from Gillespie's method from its first principle. Since the RRM algorithm is equivalent to to PDM method, the PDM method is also an exact stochastic simulation method. 

We also introduced an equivalent, but simplified version of the RRM, called the marginal Reduced Rate Method (mRRM). Even though this algorithm is easier to implement, it still retains all the benefits of the RRM algorithm.
The RRM and mRRM algorithms have several advantages compared to PDM:
\begin{itemize}
    \item RRM is considerably faster in a crowded population.
    In contrast with PDM, in the RRM and mRRM algorithms the selected reaction is always executed due to the different formulation of the propensities.
    This is a particularly important feature in cell culture models, as the immotile cells or the ones that are unable to proliferate do not slow down the simulation time.
    Moreover, this 'seemingly unnecessary' population of cells may still engage in other important processes needed for the development of the whole population.
    \item Furthermore, as a result of this, the probability of the events is known in every iteration of the algorithm. Thus, giving a probabilistic interpretation of different emerging phenomena is more straightforward in a simulation done with the mRRM algorithm (see Sec. \ref{sec:toy}).
    \item Since the design of the RRM and mRRM algorithm is based on the classification of cells according to their free adjacent sites, including new classes to this method is straightforward. 
    For example, one may include several cell types, or even extend the model incorporating the fact that the chemical substances $S_1,\dots, S_\omega$ may influence the state of the cells.
    This could be achieved in the PDM method only in very inconvenient ways, if it is possible at all.
\end{itemize}

\section*{Acknowledgement}
I would like to thank dr. Tam\'as \'Arp\'adffy-Lovas for proofreading the text and significantly improving the quality of the article. I would like to thank dr. J\'anos Benke and dr. Fanni K. Ned\'enyi for having a discussion about my proof, dr. G\'abor Sz\H{u}cs for his comments about the terminologies, Norbert Bogya for his assistance with the LaTeX template and Andr\'as Lonczkor for his ideas about the implementation of the algorithms.

\end{document}